\documentclass[11pt]{article}

\usepackage{amsfonts}
\usepackage{amssymb}
\usepackage{amsmath}
\usepackage{cite}
\usepackage{color}
\usepackage{esint}
\usepackage{epsfig}
\usepackage{graphicx}
\usepackage{graphics}
\usepackage{latexsym}
\usepackage{mathrsfs}
\usepackage{paralist}
\usepackage{psfrag}
\usepackage{subfigure}
\usepackage{wrapfig}
\usepackage{upgreek}
\usepackage{times} 
\usepackage{xspace}
\usepackage{setspace}

\usepackage{stmaryrd}
\usepackage{wasysym}
\usepackage{algorithm}
\usepackage{algorithmic}

\input epsf

\newtheorem{theorem}{Theorem}

\newtheorem{definition}{Definition}
\newtheorem{lemma}[theorem]{Lemma}

\def\ie{\textit{i.e.}\xspace}
\def\etal{\textit{et al.}\xspace}

\def\eg{\textit{e.g.}\xspace}

\newcommand{\myproof}{\noindent{\sc Proof.~}}
\newcommand{\eop}{\hfill\usebox{\smallProofsym}\bigskip}
\newsavebox{\smallProofsym}
\savebox{\smallProofsym}
	{
	\begin{picture}(6,6)
	\put(0,0){\framebox(6,6){}}
	\put(0,2){\framebox(4,4){}}
	\end{picture}
	}
\newenvironment{proof}{\myproof}{\eop}

\def\boxdist{\ell}
\def\dist{{\mathbf{d}}}
\def\rspace{{\mathbb R}}
\def\dspace{{\mathbb R}^d}

\def\weight{{\omega}}
\def\path{\Pi}
\def\Cone{{\cal C}}
\def\RCone{{\cal C}}

\def\Ang{{\Phi}}  

\def\mybox{{\mathbb B}}
\def\Successor{{\boxminus}}
\def\father{{\mathbb P}}
\def\size{{\vartheta}}

\def\Nbr{{\cal N}}

\def\Rep{{\cal R}}

\DeclareMathOperator{\EMST}{EMST}
\DeclareMathOperator{\MST}{MST}
\DeclareMathOperator{\DisjointCrossingEdge}{DCE}
\DeclareMathOperator{\AllCrossingEdge}{ACE}

\newcommand{\ignore}[1]{{}}

\def\aratio{{\beta}}
\def\udist{{\varrho_2}}  
\def\ldist{{\varrho_1}}  
\def\wdist{{\varrho}}  
\def\usize{{\varepsilon_2}}  
\def\lsize{{\varepsilon_1}}  

\def\encbox{{\APLuparrowbox}}  

\def\virtualbox{{\boxbox}} 

\def\tvbox{{\Box}} 

\makeatletter 

\title{Efficient Construction of Spanners in $d$-Dimensions}

\author{Sanjiv Kapoor\thanks{Department of Computer Science, Illinois
    Institute of Technology, Chicago, IL, 60616.
 Email: \texttt{kapoor@iit.edu}, \texttt{xli@cs.iit.edu}. The paper was first written during November 2007, and then revised during March 2008, June 2008, November 2008, March 2009, November 2009, June 2010, and May 2011.}
 \and XiangYang Li$^{*}$}

\date{May, 2, 2011}

\begin{document}

\maketitle

\vspace{-0.5in}

\begin{abstract}
In this paper we consider the problem
of efficiently constructing  $k$-vertex
 fault-tolerant geometric $t$-spanners in $\dspace$ (for $k \ge 0$ and
 $t >1$).
Vertex fault-tolerant spanners were introduced by Levcopoulus et. al in 1998.
For $k=0$, we present an $O(n \log n)$ method using the algebraic
computation tree model to find a $t$-spanner with 
 degree bound $O(1)$ and weight $O(\weight(MST))$.
This resolves an open problem.
For $k \ge 1$, 
 we  present an efficient method that, given $n$ points in $\dspace$, 
 constructs  $k$-vertex fault-tolerant $t$-spanners with the maximum
 degree bound $O(k)$ and weight bound $O(k^2 \weight(MST))$ in
 time $O(n \log n)$.
Our method achieves the best possible bounds on degree, total edge length,
 and the time complexity, and
 solves the  open problem of efficient
 construction of (fault-tolerant) $t$-spanners in $\dspace$ in time
 $O(n \log n)$.
\end{abstract}




\section{Introduction}

In this work we consider
the problem of constructing spanner graphs to approximate the complete Euclidean graph.
Given an edge weighted graph $G=(V,E,W)$, where $w(e)$ is the
 weight of an edge $e$,
 let $d_G(u,v)$ denote the shortest
 distance from vertex $u$ to vertex $v$ in graph $G$.
The weight of the graph $G$, $\weight(G)$, is
 the sum of the edge weights of edges in $G$.
A subgraph $H=(V,E')$, where $E' \subseteq E$, is called a
 $t$-spanner
 of the graph $G$, if for \emph{any} pair of vertices $u,v \in V$,
 $d_H(u,v) \le t \cdot d_G(u,v) $.
The minimum $t$ such that $H$ is a $t$-spanner of $G$ is called
the \emph{stretch factor} of $H$ with respect to $G$.
An Euclidean graph is a graph where  the vertices are points in $\dspace$
 and the weight of every edge $(u,v)$
 is the Euclidean distance $\|uv\|$ between its end-vertices $u$ and $v$.
Spanner graphs in $\rspace^d$  have  been extremely well studied.
We consider spanner graphs with additional
properties of  bounded degree, low weight and fault tolerance.


In this paper,
we  study $t$-spanners and $k$ vertex
fault-tolerant $t$-spanners ($(k,t)$-VFTS for short) for
 a set $V$ of $n$ points in $\dspace$.
A subgraph $H=(V,E)$ is $k$ vertex  fault-tolerant,
 or $k$-VFT for short,
 if for any pair of vertices $u$ and $v$ with $uv \not \in E$,
 there are $k+1$ vertex  disjoint paths from $u$ to $v$ in $H$.
Here two paths $\path_1$ and $\path_2$
 from $u$ to $v$ are said to be vertex disjoint if the
 only common vertices of $\path_1$ and $\path_2$ are $u$ and $v$.
A geometric graph $H=(V,E)$ is termed $(k,t)$-VFTS if
 for any  subset $F \subset V$ of at most $k$ vertices and any two vertices
 $w_1, w_2 \in V\setminus F$, the graph
 $H(V\setminus F,E_H')$, where $E_H'=E_H \setminus \{(u,v) \mid u \in F,
 \text{ or } v \in F \}$, contains a path $\path(w_1,w_2)$ from $w_1$
 to $w_2$
 with length at most $t\|w_1 w_2 \|$.
Given an Euclidean graph  that is  $k$ vertex
 fault-tolerant (VFT) and a real number $t>1$, the aim here is to construct
 a subgraph $H$ which is a $(k,t)$-VFTS subgraph,
 with a bounded vertex degree,  and
 a bounded weight, \ie, $\weight(H)/\weight(\MST(G))$ is bounded by a
 specified small constant, where $\MST(G)$ is
 the minimum weighted spanning tree of $G$.

A greedy algorithm
 has been used to construct spanners for various
 graphs \cite{eppstein1996sta,regev1995wgg,GLN99,GLN02,chandra92new}.
For a graph $G=(V,E)$ with  $|V|=n$ and an arbitrary edge weight,
 Peleg and Schaffer \cite{peleg1989gs} showed that any $t$-spanner
 needs at least
 $n^{1+\frac{1}{t+2}}$ edges; thus there is edge weighted graph such that
 any $t$-spanner of such a graph has weight at least $\Omega(n^{\frac{1}{t+2}}
 \weight(\MST))$ (the bound is obtained by letting the weight of each
 edge be $1$).
Chandra \etal \cite{chandra95new} showed that the greedy algorithm constructs a
 $t$-spanner of weight at most
 $(3+\frac{16t}{\epsilon^2})n^{\frac{2+\epsilon}{t-1-\epsilon}} \cdot
  \weight(\MST)$
 for every $t>1$ and any  $\epsilon >0$.
Regev \cite{regev95weight} proved that the $t$-spanner constructed by the greedy
 algorithm has weight at most $2e^2 \ln n \cdot n^{\frac{2}{t-1}}
 \cdot \weight(\MST)$ when $t
 \in [3, 2\log n +1]$, and
 has weight at most $(1+\frac{4\log^2 n + 2\log n}{t+1 -\log n})
 \cdot \weight(\MST)$ when $t > 2\log n +1$,
 by studying the girth of the constructed $t$-spanner.
Elkin and Peleg \cite{elkin-JC04} recently showed
  that for any constant $\epsilon, \lambda > 0$ there exists a
  constant $\beta = \beta(\epsilon, \lambda)$ such that for every
 $n$-vertex graph $G$ there is an efficiently constructible $(1+
 \epsilon, \beta)$-spanner of size $O(n^{1 + \lambda})$.

Constructing $t$-spanners
 \cite{Keil88,RS,S91,vaidya1991sga,das1993oss,AS94,AS94-J,das1994fac,das-soda95,arya95euclidean,chandra95new,RS98,AS99,GLN99,narasimhan00,karavelas01,NS02,BGS02,GLN02,abam-soda07}
 and $(k,t)$-VFTS
 \cite{CZ03-faulttolerantspanner,LNS98,lukovszki99,CL00} for Euclidean graphs
 has been extensively studied in the literature.
For computing $t$-spanners of $O(1)$ degree and $O(\weight(\EMST))$
 weight, the current best result \footnote{Private communication with
 M. Smid.} using algebraic computation tree model is a method with
 time complexity $O(n \log^2 n / \log\log n)$.
An $O(n\log n)$ algorithm
 which uses an algebraic model together with indirect addressing
 has been obtained in \cite{GLN02}.
While this model is acceptable in practice, the problem of computing
low weight spanners in the algebraic decision tree model in time
$O(n\log n)$ is still open.
We resolve this problem and  extend the techniques introduced in
the first part of the paper to allow us to compute the $k$-fault-tolerant
spanners efficiently.

In this paper we will also consider constructing
 $(k,t)$-VFTS for $k  \ge 1$ for
 the complete Euclidean graphs on $n$ points $V$ in $\dspace$.
The problem of constructing $(k,t)$-VFTS  for Euclidean graphs was
 first introduced  in \cite{LNS98}.
Using the well-separated pair decomposition \cite{callahan1995dmp},
 Callahan and Kosaraju
  showed that a $k$-VFT spanner can be constructed
 (1) in $O(n \log n +
  k^2 n)$ time with $O(k^2 n)$ edges, or
(2) in $O(n k \log n)$ time with
  $O(k n \log n)$ edges, or
 (3) in time $O(n \log n + c^k n)$ with degree $O(c^k)$ and total
  edge length $O(c^k \cdot \weight(\EMST))$.
Here the constant $c$ is independent of $n$ and $k$.
Later, Lukovszki \cite{lukovszki99} presented a method to construct a
$(k,t)$-VFTS with  the asymptotic optimal number of edges $O(kn)$ in
 time $O(n \log ^{d-1} n + n k \log \log n)$.
Czumaj and Zhao \cite{CZ03-faulttolerantspanner}
 showed  that there are Euclidean graphs such that
 \emph{any} $(k,t)$-VFTS has weight at least  $\Omega( k^2)
 \weight(\EMST)$,
 where EMST is the Euclidean minimum spanning tree connecting $V$.
They then proved that
 one can construct a $(k,t)$-VFTS using  a greedy method
 \footnote{Edges
   are processed in increasing order of length and an edge $(u,v)$ is
   added only if $H$ formed by previously added edges does not have
   $k+1$ internally node-disjoint  paths connecting $u$ and $v$,
  each with length at most $t   \|u-v\|$.}
 for a set $V$ of
 $n$ nodes, that has maximum degree $O(k)$ and total edge length
 $O(k^2 \weight(\EMST))$.
However it is unknown, given arbitrary $k$,
  an Euclidean graph and a pair of vertices $u$ and $v$,
 whether we can determine in polynomial time if there are $k+1$ vertex-disjoint
 paths connecting them and each path has a length at most a given
   value $t\|uv\|$.
Notice that this problem is NP-hard when we are given
 a  graph $G$ with arbitrary weight function.
Czumaj and Zhao further presented a method to construct a $(k,t)$-VFTS
   for Euclidean graphs in time
 $O(n k \log ^d n + nk^2 \log k)$ such that it has the maximum node
   degree $O(k)$
 and  total edge length $O(k^2 \log n ) \cdot \weight(\EMST)$ for
   $k>1$.
Observe  that there is a gap between the lower bound
  $O(k^2 ) \cdot \weight(\EMST)$ and
   the achieved upper bound $O(k^2 \log n ) \cdot \weight(\EMST)$
 on the total edge length.

\paragraph{Our Results:}
The contributions of this paper are as follows.
In  the first part of the paper,
 given a set $V$ of points (such input points are called \emph{nodes}
 hereafter) in $\dspace$ and an arbitrary real number $t >1$,
 we present a method that
 runs in time $O(n \log n)$ using the algebraic computation tree model
 and constructs a $t$-spanner graph whose
 total edge length is $O(\weight(\EMST))$. The hidden constants depend on
 $d$ and $t$, or more precisely, the number of cones used in our
 method, which is $O((\frac{1}{t-1})^d)$.
This solves an open question of finding a method with
 time-complexity $O(n \log n)$ in the algebraic computation tree
 model.
The main techniques used in our methods are listed below.
\begin{compactenum}
\item
We first apply a special well-separated pair decomposition, called
 \emph{bounded-separated pair decomposition} (BSPD) which is produced
 using a  split-tree partition of input nodes
 \cite{callahan1995dmp}. The split-tree partition uses boxes that tightly
enclose a set of nodes, i.e., each side of the box contains a point
 from $V$. In our decomposition, we need to ensure that
every pair $(X,Y)$ of separated sets of nodes is contained in two,
 almost equal sized boxes, $b(X)$ and $b(Y)$, respectively, where
 $b(X)$($b(Y)$ respectively)
contains only the node set $X$($Y$ respectively).
These boxes are  termed {\em floating virtual boxes} since they can be positioned
in a number of ways.
An important property of the BSPD that we construct
 is that for every pair  of nodes sets $(X, Y)$ in
the  decomposition, the distance between $b(X)$ and $b(Y)$ is not only
not too small (these conditions are from WSPD), but also not too
 large, compared
with the sizes of  boxes containing them respectively.

\item
To facilitate the proof that the structure constructed by our method
 is a $t$-spanner we use neighborhood cones. At every point $x$ we use
a cone partition of the space around $x$ by a set of basis vectors.
To guide the addition of spanner edges
we introduce the notion of {\em General-Cone-Direction}
for a pair of boxes $b$ and $b'$. This notion ensures that every point
 $x$ contained inside the box
$b$ has all the nodes inside $b'$ within a collection of
cones ${\cal C}$, each cone with apex $x$.
The angular span of the cones ${\cal C}$ is bounded from above by some
 constant (depending on the spanning ratio $t$), which ensures the
 spanner property.

\item
To prove that the structure constructed by our method
has low-weight, we introduce the \emph{empty-cylinder property}.
A set of edges $E$ is said to have the empty-cylinder property if
 for every edge $e \in E$, we can find an empty cylinder (that does
 not contain any end-nodes of $E$ inside) that uses a
 segment of $e$ as its axis and has radius and height at least some
 constants factor of the length of $e$.
We  prove that a set of edges $E$ with empty-cylinder property and
 empty-region property  has a total weight proportional to the minimum
 spanning tree
of the set of end-nodes of edge $E$.
\end{compactenum}

In the second part of the paper,
given $V$ in $\dspace$, $t>1$ and a constant integer $k> 1$,
 we present a method  that
 runs in time $O(n \log n)$ and constructs a $k$-vertex fault-tolerant
 $t$-spanner graph with following properties:
 (1) the maximum node degree is $O(k)$, and (2) the total edge length is
 $O(k^2)\weight(\EMST)$.
 This achieves an optimal weight bound and degree bound
of the spanner graph, which is the first such  result known in the literature.
The second part utilizes techniques introduced in the first part.

The paper is organized as follows.
We present our  method of
constructing $t$-spanner in time $O(n\log n)$ in Section \ref{sec:boxtree},
 and prove the properties of the structure and study the time
 complexity of our method in Section \ref{sec:proofs-k1}.
In Section \ref{sec:boxtree-k}, we present and study our method of
 constructing $(k,t)$-VFTS.
We conclude our paper in section \ref{sec:conclusion}.




\section{$t$-Spanner in $\dspace$ Using Compressed Split-tree}
\label{sec:boxtree}

In this section, we present an efficient method in the algebraic
 computation tree model with time complexity
 $O(n \log n)$ to construct a  $t$-spanner for any given set of nodes
 $V$ in $\rspace^d$ for any $t>1$.
Our method is based on a variation of the compressed split-tree
 partition of $V$: We partition all pairs of nodes
 using  a \emph{special}
 well-separated pair decomposition based on a variant of  split-tree
 that uses boxes with bounded aspect ratio.

\subsection{Split Tree Partition of a set of Nodes}

We use $\dist(x,y)$ to denote the distance between
 points $x $ and $y$ in $\dspace$ in the $L_p$-metric for $p \ge 1$.
We will focus on Euclidean distance here.
We define our partition of input nodes $V \in \dspace$ using a Compressed
Split-tree, a structure first used in \cite{Vai1}.
Let $x_i$ be the $i$th dimension in $\dspace$ and $x=(x_1, x_2,
\cdots, x_d)$ be a point in $\dspace$.
Then an orthogonal
 box $b$ in $\dspace$ is $\{x=(x_1, x_2, \cdots, x_d) \mid L_i \le
x_i \le R_i\}$, where $L_i < R_i$, $i=1, 2, \cdots, d$, are given values
defining the bounding planes of the box.
Given a box $b$ in $\dspace$ we define the following terminology:
\begin{compactitem}
\item
$|b|$ is the number of nodes from $V$ contained in the box $b$.
\item
$\dist(b_1, b_2)$ is the Euclidean distance between
the boxes $b_1$ and $b_2$, i.e., $\min_{x \in b_1, y\in b_2} \|x-y\|$.
\item
For a box $b$, $\size(b)$ denotes the \emph{size} of box $b$, i.e.
the length, $\max_{1 \le i\le d}(R_i -L_i)$, of the longest side of $b$.

\item
The aspect ratio of a box $b$ is defined as the ratio of the longest
side-length over the smallest side-length, i.e., $\max_{1 \le i,j\le
d}(R_i -L_i)/(R_j -L_j)$.
\end{compactitem}

Given a point set, $S$, we will refer to the smallest orthogonal box enclosing
the point set $S$ as $b=\encbox(S)$.
Such a box $b$ is called \emph{enclosing-box}  of
a point set $S$ hereafter.
Here the enclosing-box $b$ does not necessarily have a good aspect ratio.

\begin{definition}[Tight-Virtual Box]
\label{definition:tight-virtual}
Given a  box $b=\encbox(S)$
 we define a box  $\tvbox(b)$ as a {\em tight-virtual box} if it has the
following properties:

\begin{compactenum}

\item
$\tvbox(b) \supseteq b$, i.e. it contains $b$ inside

\item
longest side of $\tvbox(b)$ is exactly $\size(b)$, the longest side of
box $b$.

\item
$\tvbox(b)$ has an aspect ratio at most a constant $\aratio \le 2$.

\end{compactenum}
\end{definition}

Given a set $V$ of $n$ $d$-dimensional nodes,
 let $\delta(V)$ be the smallest pairwise distance between all pairs
 of nodes in $V$.
We next define a special split-tree similar to the structure defined
 in~\cite{Vai1,callahan1995dmp}.

\begin{definition}[Compressed split-tree]
\label{definition:Compressed split-tree}
A compressed split-tree,
termed $CT(V)$,  is a rooted tree of
 $d$-dimensional boxes defined as follows:
\begin{compactenum}
\item Each vertex $u$ in the  compressed split-tree is mapped to
 a $d$-dimensional box $b=\mybox(u)$, and associated with a
 tight-virtual box $\tvbox(b)$.

\item The root vertex, termed {\em root}, of the tree $CT(V)$ is
associated with
the enclosing-box $\mybox(root)=\encbox(V)$
containing all the nodes in $V$.
Associated with this box is a tight-virtual box
$\tvbox( \mybox(root))$ which has a bounded aspect
ratio $\aratio \le 2$ enclosing the box $\mybox(root)$.

\item Each internal vertex $u$ (associated with a box $b=\mybox(u)$
and the tight-virtual box $\tvbox(b)$) in the tree $CT(V)$
 has two children vertices, if $b$ contains at least $2$ nodes from $V$.
Consider the two boxes, $B_1', B_2'$, obtained by subdividing
 $b$  into $2$   smaller boxes
 ${b'}_i$, $ 1\le i \le 2$,
 cutting $b$ by a hyperplane passing through
 the center of $b$ and perpendicular to the longest
  side of $b$.
Shrink $B'_i, 1 \leq i \leq 2$ to obtain minimum sized enclosing-box $b_i$,
  containing the same set of nodes as $B_i'$, \ie,
  each face of $b_i$ contains a node of $V$.
  Let $\Successor(b) = \{b_1, b_2 \}$.
With each box $b_i$ in $\Successor(b)$, we  associate a  tight-virtual
 box  $\tvbox(b_i)$ with an aspect ratio at most $ \aratio \le 2$.
Then the children vertices of the vertex $u$ are two boxes $b_1, b_2$.
Additionally, $\tvbox(b_1)$ and $\tvbox(b_2)$ are disjoint and are
 contained inside $b$.

\item There is  a tree edge from $b$ to every $b_i \in  \Successor(b)$.
Notice that  neither $b_1$ nor $b_2$ is empty of nodes inside.
The box $b$ from which $b_i$ is obtained by this procedure is referred to as
 the \emph{father} of the $b_i$, denoted as $\father(b_i)$.
A box $b$ that contains only one node is called a \emph{leaf  box}.
For simplicity of presentation, we assume that any leaf box has a size
 $\epsilon$ for sufficiently small $0 < \epsilon \ll \delta(V)$.

\item
The \emph{level} of a box $b$ is the number (rounds) of subdivisions used to
produce $b$  from the root box.
The level of the box $\mybox(root)$ is then $0$.
If a box $b$ has level $j$, then each box in $\Successor(b)$ has
 level $j+1$.
\end{compactenum}
\end{definition}
In Lemma \ref{lemma:finding-tight-virtual-box},
 we will show that the tight-virtual boxes $\tvbox(b_i)$, $i=1,2$, can be
constructed from $\tvbox(b)$ in $O(d)$ time

The tree $CT(V)$ is called a \emph{canonical partition}
 split-tree of $V$.
One difference between our structure and  the split-tree
 structure used in \cite{callahan1995dmp} is that we associate with each
 box $b$ in $CT(V)$ a tight-virtual box $\tvbox(b)$, while in
 \cite{callahan1995dmp} different boxes are used.
Another major difference is the \emph{floating-virtual-boxes} to be
 introduced later.

\begin{lemma}
\label{lemma:finding-tight-virtual-box}
Given a box $b$ and its associated tight-virtual box $\tvbox(b)$, we
can find the tight-virtual box $\tvbox(b_i)$ for each children box
$b_i \in \Successor(b)$ in $O(d)$ time.
\end{lemma}
\begin{proof}
Obviously, for the root box $\encbox(V)$, we can find a tight-virtual
 box $\tvbox(\encbox(V))$ in $O(d)$ time.
We now show that given a vertex $b=\encbox(S)$ (enclosing-box of some
 subset $S$ of
 nodes in $V$) in $CT(V)$
 and a child enclosing-box $b_1$ obtained by subdividing $b$ by a hyperplane
 $h$, we can construct a tight-virtual box $\tvbox(b_1)$ from
 the tight-virtual box $\tvbox(b)$ efficiently as follows.
Let $\tvbox(b)_h$ be the box, which contains $b_1$,  obtained by partitioning
 $\tvbox(b)$ using the hyperplane $h$.
Assume w.l.o.g that $b_1$ is the one located with the same center as
$\tvbox(b)_h$.
Box $\tvbox(b)_h$ can now be shrunk as $b_s$ until one of
 the sides of the shrunk box $b_s$ meets a side of $b_1$.
Other sides of $b_s$ which are larger than $\size(b_1)$
 can be shrunk to meet $b_1$ if the aspect ratio is not below $\aratio$.
This gives us $\tvbox(b_1)$. The aspect ratio
 of $\tvbox(b_1)$ is bounded by the aspect ratio of $\tvbox(b)_h$ if
 $\aratio \le 2$.
\end{proof}

For the purpose of constructing a spanner in $\dspace$, we introduce
 another box, called \emph{floating-virtual box}, associated with each
 box in the tree $CT(V)$.
For box $b$, let $\tvbox(\father(b))_h$ be one of the two boxes
 that is produced by halving
 the (longest dimension of) tight-virtual box $\tvbox(\father(b))$  and that
 contains $b$ inside.
Since  the tight-virtual box $\tvbox(\father(b))$ has an aspect ratio
 bounded by $ \aratio \le 2$,
 $\tvbox(\father(b))_h$ has an aspect ratio  bounded by $ \aratio \le 2$ also.

\begin{definition}[Floating-Virtual Box]
\label{definition:floating-virtual-box}
Consider a compressed split-tree $CT(V)$ for a set of input nodes $V$.
For a box $b$, a box, denoted as $\virtualbox(b)$, is
 termed as {\em a floating-virtual box associated with
 the box $b$} if the following properties hold:
\begin{compactenum}
\item it includes the tight-virtual box $\tvbox(b)$ of the box $b$ inside,
\item it is contained inside the parent box $\father(b)$ of $b$,
\item it has an aspect ratio at most a constant $\aratio \le 2$, and
\item \label{fv-in-p} it is contained inside the box
  $\tvbox(\father(b))_h$, halved
  from the tight-virtual box $\tvbox(\father(b))$.
\end{compactenum}
\end{definition}

It is worth to emphasize that, for a box $b$, a floating-virtual box
 to be used by our method
 is not unique: it also depends on some other boxes to be paired with.
It is also easy to show that the floating-virtual boxes of two
 disjoint boxes $b_1$ and $b_2$ will be always disjoint because of the
 property~\ref{fv-in-p} in
 Definition~\ref{definition:floating-virtual-box}.
Table~\ref{tab:notation} summarizes some of the notations used in the
 paper.
See Figure~\ref{fig:boxes} for illustration of some concepts defined
 in this paper.

\begin{table*}[tbhp]
\begin{center}
\caption{Notations and abbreviations used in this paper.}
\label{tab:notation}
\begin{tabular}{p{1in}|p{5in}}
\hline
\hline
$\encbox(S)$ & the enclosing-box of a set of nodes $S$\\
$\Successor(b)$ & the two minimum sized boxes produced by halving $b$ and
then shrinking the produced boxes to be enclosing-boxes\\
$\tvbox(b)$ &  the tight-virtual box that contains a box $b$ inside and has
an aspect ratio $\le \aratio \le 2$.\\
$\virtualbox(b)$ & a floating-virtual box  containing the tight-virtual box
$\tvbox(b)$ inside ($\tvbox(b)$ and $\virtualbox(b)$ may be same) and
has an aspect ratio $\le \aratio$. Here $\virtualbox(b)$ is not
unique in our algorithm: it depends on the box $b'$ to be paired with
for defining edges.\\
$\size(b)$ & the size of the box $b$, \ie, the length of the longest
side.\\
$\dist(b_1, b_2)$ & the Euclidean distance between two boxes $b_1$,
$b_2$.\\
$\boxdist(b_1, b_2)$ & the edge-distance   between two boxes $b_1$,
$b_2$. This is equal to $\dist(b'_1,b'_2)$ where $b'_1$ and $b'_2$ are
the floating-virtual boxes for bounded-separated boxes $b_1$ and $b_2$.\\
$\father(b)$ & the parent vertex of a box $b$ in the tree $CT(V)$\\
\hline
\end{tabular}
\end{center}
\end{table*}

\begin{figure} [htpb]
\begin{center}
\epsfxsize=4.5in\epsfbox{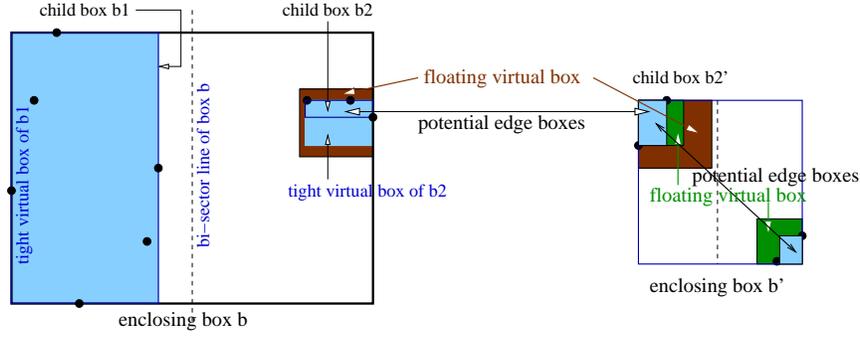}
\end{center}
\caption{An illustration of several concepts defined in this
  paper. Here for a box $b_2'$, depending on the pairing box, the
  floating-virtual box for $b_2'$ could be different. When the box $b_2'$
  is paired with the box $b_2$, the floating-virtual box
  $\virtualbox(b_2')$ is shaded as  brown in the figure. When  the box $b_2'$
  is paired with the box $b_1'$, the floating-virtual box
  $\virtualbox(b_2')$ is shaded as  green in the figure.}
\label{fig:boxes}
\end{figure}

Observe that the compressed split-tree proposed here is slightly
 different from the  split-tree defined in \cite{callahan1995dmp}: we
 define the tight-virtual boxes and also associate a tight-virtual box
 with some  floating-virtual boxes: these floating-virtual boxes will be
 determined by a procedure to be described later.
Given that a  split-tree  can be constructed in $O(n
 \log n)$ steps \cite{callahan1995dmp},  it is easy to show the
 following theorem (Theorem \ref{theo:timing-cv}). The proof is similar to the
 proof in \cite{callahan1995dmp} and is thus omitted.

\begin{theorem}\label{theo:timing-cv}
The compressed split-tree $CT(V)$ can be constructed in time $O(d n \log
n)$ for a set of nodes $V$ in $\dspace$.
\end{theorem}

In the rest of the paper, we will mainly focus on the tight-virtual
 boxes $\tvbox(b)$ for all enclosing-boxes $b$ produced in the
 compressed split-tree. For ease of description,
we will refer to $b$ and $\tvbox (b)$ by $b$ itself.
Thus $ \virtualbox(\tvbox(b))$ will refer to the same box as
 $\virtualbox(b)$.
The difference here is that $\tvbox(b)$ has an aspect ratio $\le
 \aratio$ while the aspect ratio of an enclosing-box $b$ could be
 arbitrarily large.

\subsection{Well-Separated Pair Decomposition and Bounded-Separated
  Pair Decomposition}

Our method of constructing a spanner efficiently
 will use some decomposition of all pairs of nodes similar to
 well-separated pair decomposition (WSPD)~\cite{callahan1995dmp}.

\paragraph{Well-Separated Pair Decomposition (WSPD):}
Recall that, given two sets of points $A,B \in \dspace$, a set
${\cal R}(A,B)=\{ \{A_1, B_1 \},  \{A_2, B_2 \},\cdots,  \{A_p, B_p \} \}$ is
 called a \emph{well-separated realization} of the interaction product
 $A \otimes B =\{ \{x,y\} \mid x\in A, y \in B, \text{ and } x \not = y\}$ if
\begin{compactenum}
\item $A_i \subseteq A$, and $B_i \subseteq B$ for all $i \in [1,p]$.
\item $A_i \cap B_i = \emptyset$  for all $i \in [1,p]$.
\item  $(A_i \otimes B_i) \cap (A_j \otimes B_j) = \emptyset$ for all
 $ 1 \le i < j \le p$.
\item $A \otimes B = \bigcup_{i=1}^{p} {A_i \otimes B_i}$.
\item $A_i$ and $B_i$ is \emph{well-separated}, \ie, the distance
$\dist(A_i, B_i) \ge \wdist \max(\size(A_i), \size(B_i))$ for some
constant $\wdist$,
 where $\size(X)$ of a point set $X$ is the radius of the smallest disk
containing $X$, which is of the same order of the size of the smallest
box $\encbox(X)$ containing $X$.
\end{compactenum}
Here $p$ is called the \emph{size} of the realization ${\cal R}(A,B)$.

Given the split-tree $CT(V)$, we say that a well-separated realization of
 $V \otimes V$ is a \emph{well-separated pair decomposition} (WSPD) of $V$
 based on $CT(V)$, if for any $i$, $A_i$ and $B_i$  are
 the sets of nodes contained in some  enclosing-boxes, $\encbox(u_i), u_i \in CT(V)$
and $\encbox(v_i), v_i \in CT(V)$, respectively.
Here we overuse  notations a little bit, for a box $b \in CT(V)$, we
 also refer via $b$ as the subset of nodes from $V$ contained inside $b$.
In \cite{callahan1995dmp}, it is  shown that a well-separated
 pair decomposition based on a split-tree $CT(V)$ can be constructed
 in linear time  $O(n)$ when the split-tree $CT(V)$ is given.

\begin{theorem}\cite{callahan1995dmp}
\label{theo:linear-time}
Given the  split-tree $CT(V)$, we can construct a
 well-separated pair decomposition based on $CT(V)$ in linear-time
 $O(n)$ and the realization of this WSPD has size $O(n)$.
\end{theorem}
\begin{proof}
For completeness of presentation, we briefly review the proof here.
The algorithm itself is recursive.
In this proof, when we mention a box $b$, we always mean the
  tight-virtual box $\tvbox(b)$
 corresponding to the enclosing-box used in the split-tree.
Each vertex  of the compressed split-tree $CT(V)$ has
 $2$ children enclosing-boxes, containing point sets denoted by  $A_1, A_2$.
We construct a well-separated realization for  $A_1 \otimes A_2$.

We now focus on how to construct a realization for a pair of boxes $b$
 and $b'$.
If $(b,b')$ is well-separated, \ie,
 $\dist(b,b') \ge \wdist \max(\size(b), \size(b'))$, then we are
 done for this pair.
Otherwise, for simplicity, we assume that $\size(b) \ge \size(b')$.
In this case, let $\{c_i \mid 1 \le i \le 2\}$  be the $2$
  children enclosing-boxes  of $b$ in the tree $CT(V)$.
We then recursively construct a well-separated realizations of $c_i
 \otimes b'$, for each child box $c_i$ of $b$
 and then return the union of these realizations as the well-separated
 realization for the pair $(b,b')$.
Based on the compressed split-tree $CT(V)$, we will
 have a WSPD-computation tree $T$ as follows:
\begin{compactenum}
\item
 The root of $T$ is $( b_0, b_0)$, where $b_0$ is the box
 corresponding to the root vertex of $CT(V)$;
\item
For each node $(b,b')$ in $T$, if $b=b'$, then it has the following
children $(b_i, b_j)$ where $1 \le i \le j \le 2$ and $b_1$, $b_2$
 are   children enclosing-boxes  of $b$ in the compressed split-tree $CT(V)$.
However, observe that here $(b_i,b_i) \in T$ for each children box $b_i$ of
$b$.

\begin{compactenum}
\item If $\dist(b,b') \ge \wdist \max(\size(b), \size(b'))$, then it does
 not have any children since $(b,b')$ is a pair of well-separated boxes.
\item If $\dist(b,b') < \wdist \max(\size(b), \size(b'))$ and
 $\size(b) \ge \size(b')$, then it has  $2$  children
 nodes $(c_i, b')$ where $c_i$ is a tight-virtual child box of $b$ in the
 tree $CT(V)$.
Recall that here $c_i$, $1\le i \le 2 $ is produced by halving the
 longest dimension of $b$ and then shrinking the corresponding boxes into
 the smallest tight-virtual boxes with aspect ratio at most $\aratio$.
\item If $\dist(b,b') < \wdist \max(\size(b), \size(b'))$ and
 $\size(b) < \size(b')$, then it has  $2$  children
 nodes $(b, c_i)$ where $c_i$ is a tight-virtual children box of $b'$ in the
 tree $CT(V)$.
\end{compactenum}
\end{compactenum}
For simplicity of analysis, for any vertex $(b,b')$ in the
 WSPD-computation tree $T$, we always assume that $\size(b) \ge
 \size(b')$; otherwise, we reorder them and rename them.
A careful analysis in  \cite{callahan1995dmp} show that
 the WSPD-computation tree $T$ has size at most
 $2n \cdot 2^{2} + 4n \cdot (4\ldist)^d  = O(n)$ vertices.
The theorem then follows.
\end{proof}

\paragraph{Bounded-Separated Pair Decomposition (BSPD):}
Observe that a pair of boxes $(b,b')$ in WSPD may have a distance
 arbitrarily larger than $\max(\size(b), \size(b'))$, especially, the
 WSPD produced in \cite{callahan1995dmp}.
To produce spanners with low weight, we do not want to include
 arbitrarily long edges in the spanner, unless it is required.
To capture such a requirement, we propose a new concept,
 a \emph{bounded-separated pair decomposition} (BSPD).
A  well-separated pair decomposition (WSPD) based on a compressed
 split-tree $CT(V)$
 is called a \emph{bounded-separated pair decomposition},
 if  it has the following additional
 property:  each pair of tight-virtual boxes $(b,b')$
 in this WSPD satisfies the property of \emph{bounded-separation of
 floating-virtual boxes}, \ie, there is a pair of floating-virtual boxes
 $b_2 =\virtualbox(b)$ and $b_2'=\virtualbox(b')$ and $(b_2,
 b_2')=(\virtualbox(b), \virtualbox(b'))$
 has \emph{bounded separation}.

\begin{definition}[Bounded Separation]
A pair of boxes $(b_2,b_2')$  has the property of
\emph{bounded-separation} if it satisfies the following properties
\begin{compactenum}
\item  
{\em Almost Equal-size Property:} $\lsize \size(b_2') \leq \size(b_2) \leq \usize \size(b_2')$.
Here we typically choose constants $\usize= 1/\lsize=2$.
The two boxes $b_2$ and $b_2'$ are called almost-equal-sized.
 \item \emph{Bounded-Separation Property:} $\ldist \max( \size(b_2),
 \size(b_2')) \le
 \dist(b_2, b_2') \le \udist \max(\size(b_2), \size(b_2'))$
 for constants $1 < \ldist < \udist$ to be specified later.
\end{compactenum}
\end{definition}

Two boxes $b$ and $b'$ present in the bounded-separated pair decomposition
are called a pair of \emph{bounded-separated boxes} in $CT(V)$.
The choice of the constants $\udist > \ldist \ge \frac{2 \sqrt {d} t}{t-1}$
 depends on the spanning ratio $t>1$ required.
The constants $\udist$ and $\ldist$ will be chosen as specified later
 to ensure the existence
 of a pair of bounded-separated boxes.
Observe  the fact that a pair of boxes $b$ and $b'$ is in BSPD
 does not imply that the distance between $b$ and $b'$ is in the range
 $[\ldist  \cdot \max(\size(b), \size(b')), \udist  \cdot
 \max(\size(b), \size(b'))]$; it is possible
 that the distance $\dist(b,b') > \udist \cdot \max(\size(b),
 \size(b'))$.
However, observe that $\dist(\father(b),\father(b')) \le \udist \cdot
\max(\size(\father(b)),  \size(\father(b')))$ for the parent boxes
$\father(b)$ and $\father(b')$ of $b$ and $b'$.


Given the split-tree $CT(V)$, we then briefly discuss how to connect
pairs of nodes to form edges in the spanner.
Our method is based on the concept of \emph{potential-edge-boxes}.

\begin{definition}[potential-edge-boxes]
Two enclosing-boxes $b$ and $b'$ in the compressed split-tree
 $CT(V)$ (corresponding to  two vertices in
 $CT(V)$) are said to be a pair of \emph{potential-edge boxes},
 denoted as $(b,b')$, if
\begin{compactenum}
\item
The  pair of floating-virtual boxes $b_2=\virtualbox(b)$ and
 $b_2'=\virtualbox(b')$ have the property of bounded-separation;
the pair of boxes $b_2$ and $b_2'$ is called
 the pair of bounded-separated floating-virtual boxes defining the
 pair of potential-edge
 boxes $b$ and $b'$.
\item
 None of the pairs $(b, \father(b'))$, $(\father(b),
 b')$,  $(\father(b), \father(b'))$, is a pair of potential-edge boxes in
 $CT(V)$.
\end{compactenum}
\end{definition}

For a pair of potential-edge boxes $(b,b')$, we use $(\virtualbox(b),
\virtualbox(b'))$ to denote the floating-virtual boxes of $b$ and $b'$
respectively that define the pair of potential-edge boxes.
From the definition of floating-virtual boxes, we have  the following lemma.
\begin{lemma}
For a pair of potential-edge boxes $b$ and $b'$ and the pair of
 floating-virtual boxes $b_2$ and $b_2'$ defining them,
 the floating-virtual box $b_2=\virtualbox(b) \supseteq b$
 cannot contain another  enclosing-box  that is disjoint of $b$,
i.e. $b'' \cap b = \emptyset \implies \virtualbox(b) \cap b'' = \emptyset $.
\end{lemma}

\begin{definition}[edge-distance of potential-edge boxes]
\label{def:edge-distance}
For a pair of potential-edge boxes $b_1$ and $b_1'$, define its
\emph{edge-distance}, denoted as $\boxdist(b_1, b_1')$,
 as the distance between the pair of
 bounded-separated floating-virtual boxes $b_2$ and $b_2'$ that define
 the pair of   potential-edge boxes $b_1$ and $b_1'$,  \ie,
 $\boxdist(b_1, b_1') = \dist(b_2, b_2')$.
\end{definition}

Given a compressed split-tree $CT(V)$ and associated boxes at each node
 we then define the \emph{edge-neighboring boxes} of an enclosing-box  $b$ as
\[
\Nbr(b)= \{ b' \mid \mbox{ boxes $b$ and $b'$ are a pair of
  potential-edge boxes in $CT(V)$}\}
\]
Observe that here the distance $\dist(b,b')$ could be arbitrarily
 larger than the maximum size of boxes $b$ and $b'$.
However, the distance $\dist(\father(b), \father(b')) \le \udist
 \max(\size(\father(b)), \size(\father(b')))$ if $b' \in \Nbr(b)$.
Given a fixed size $L$, and a tight-virtual box $b$,
 the number of tight-virtual boxes $b'$ such that (1)
 $(\virtualbox(b),\virtualbox(b'))$ has the property of
 bounded-separation, and (2) $\virtualbox(b)$ has size $L$,
 is clearly bounded by a constant.
However, this does not mean that the number of tight-virtual boxes,
 the cardinality of $\Nbr(b)$,
 that could pair with $b$ to form a pair of potential-edge boxes
 is bounded by a constant.
The reason is that the floating-virtual boxes for a tight-virtual box
 $b$ depend on with which box the box $b$ will be paired (see
 Figure~\ref{fig:boxes} for illustration).
There is an example of nodes' placement such that
 the cardinality   $|\Nbr(b)|$ could be as large as $\Theta(n)$.
The example is as follows: in $2D$, we place a node at $v_0=(0,0)$ and $n$
 nodes $v_i$ at $(0,2^{i}-\epsilon)$, for $1\le i \le n$.
An additional node $v_{n+1}$ is placed at $(0,-2^{n})$.
Here $\epsilon>0$ is a sufficiently small number.
Then the smallest box $b$ containing node $v_{n+1}$ will have $\Theta(n)$
 boxes in $\Nbr(b)$ (the sizes of these boxes are about $2^{i}$, $0\le
 i \le n-1$).

Thus, our construction method will use another set instead
\[
\Nbr_{\ge}(b)= \{ b' \mid b' \in \Nbr(b) \mbox{ and }
\size(\father(b')) \ge \size(\father(b)) \}
\]
Observe that for any pair of potential-edge boxes $b$ and $b'$,
 we either have $b' \in \Nbr_{\ge}(b)$ or  $b \in \Nbr_{\ge}(b')$, or both.
Consider $\father(b)$ and $\father(b')$.
The distance between $\father(b)$ and $\father(b')$ is at most
$\udist \max(\size(b_2), \size(b_2'))$ since
the floating-virtual box for $b$ is always contained inside $\father(b)$.
Thus there are at most $\Theta  ((\rho_2  \cdot \usize)^d)$ boxes $b'$.

\begin{lemma}
The cardinality of $\Nbr_{\ge}(b)$ is bounded by a  constant
$ \Theta((\rho_2 \cdot \usize)^d)$.
\end{lemma}
We now show that we can construct a linear size BSPD based on $CT(V)$
 in linear time.

\begin{lemma}
\label{lemm:finding-almost-equal-sized-box}
Given any pair of tight-virtual boxes $b$ and $b'$
in a BSPD with
$\size(b') \le \size(b) \le \size(\father(b'))$, in time $O(d)$,
 we can  find a
 floating-virtual box $b''$ inside $\father(b')_h$ of
 \emph{almost-equal-size} with $b$.
\end{lemma}
\begin{proof}
Notice that the size of $\father(b')_h$ is at least $\size(b) /2$
 since $\size(b) \le \size(\father(b'))$.
Pick any dimension, say $d_1$. Projecting  $\father(b')_h$ on this dimension
 results in a segment, say $xy$.
Let the segment $ab$ be the projection of $b'$ in the dimension of $d_1$.
Note that $[a,b] \subset [x,y]$.
We then align a floating-virtual box $b''$
( s.t. $\size(b)/2 \leq \size(b'') \leq \size(b)$)
 such that its projection on the dimension that contains
  $[c,d]$ starts at  $x$.
If $[c,d]$ contains $[a,b]$ we are done. Otherwise,
align $b''$ such that $d=b$. Since $b-x \geq d-c$ and $d-c \geq b-a$, the
alignment of the box $b''$ in dimension $d_1$ is possible.
This can be repeated for all dimensions.
It is easy to show that the size of the floating-virtual box $b''$ is
 at most $\size(b)$, and at least $\size(b) /2$.
This finishes the proof.
\end{proof}

\begin{theorem}
\label{theo:potential-edge-boxes}
Given the compressed split-tree $CT(V)$, we can construct a
 bounded-separated pair decomposition (BSPD) (using constants
 $\ldist$ and $\udist$)
 in linear-time $O(n)$ and the realization has size $O(n)$.
The constants $\ldist$ and $\udist$ are related to $\varrho$ of the WSPD as follows:
\begin{equation}
\begin{cases}
\ldist \le \varrho - \sqrt{d}\\
\udist \ge 2\varrho + 4 {\sqrt  d}
\end{cases}
\end{equation}
\end{theorem}
\begin{proof}
We will prove this based on  the well-separated pair decomposition
 computed in the proof of Theorem \ref{theo:linear-time}.
In the proof, we will mainly focus on the tight-virtual boxes, instead
 of the actual enclosing-boxes.
Observe that in the WSPD computation tree $T$ (defined in proof of Theorem
 \ref{theo:linear-time})
 of a WSPD based on the compressed split-tree $CT(V)$,  all the leaf
 vertices will form a WSPD.
We first build a WSPD with a constant $\varrho$, where the exact value
 of $\varrho$ will be determined later.

Consider a pair of boxes $(b,b')$ at  a leaf vertex in the
 WSPD-computation tree $T$, \ie $(b,b')$ is an element of the computed
 WSPD.
We  show that, by  adjusting the sizes of pairs of boxes
 in the WSPD computed,
 we can get a BSPD.
More specifically, for each pair of boxes $b$ and $b'$ in WSPD,
 we show how to obtain two almost-equal-sized floating-virtual boxes
 $b_2$ (containing $b$)
 and $b_2'$ (containing $b'$) such that
 $\ldist \size(b_2) \le \dist(b_2, b_2') \le \udist \size(b_2)$ for
 some constants $\ldist < \udist$.

Assume w.l.o.g.,  $\size(b') \le \size(b)$.
First of all, because of the properties of  the WSPD computation-tree $T$,
 $ \size(b) \le \size(\father(b'))$ when $(b, b')$ is a
 leaf node in the computation tree $T$.
Notice that here, to get the vertex  $(b, b')$,
  we could have split $\father(b')$ first or we could have split
 $\father(b)$ first in the WSPD computation-tree $T$.
Thus, by Lemma~\ref{lemm:finding-almost-equal-sized-box}, we can find
 a floating-virtual box $b''=\virtualbox(b')$, which  is
 almost-equal-sized to $b$ (i.e., $\size(b) /2 \le \size(b'') \le
 \size(b)$), is
 inside $\father(b')$, and contains $b'$ inside.

We now show that the distance $\dist(b, b'')$ is at least a constant
 fraction of $\size(b)$.
Obviously, \[\dist(b, b'') > \dist(b, b') - \sqrt{d} \size(b'')
 \ge \varrho \max(\size(b), \size(b')) -  \sqrt{d} \size(b)
=  (\varrho-\sqrt{d}) \size(b) \ge \ldist \size(b).\]

On the relations of $\dist(b, b'')$ and the size $\size(b)$,
 there are two complementary cases here:
\begin{enumerate}
\item
[Case 1: $\dist(b, b'') \le \udist \size(b) $:]
In this case, we have already found a pair of floating-virtual boxes $b_2=b$
and $b_2'=b''$ for the pair of boxes $b$ and $b'$ such that the distance
between the floating-virtual boxes satisfies
 that $ \ldist \size(b) \le \dist(b, b'') \le \udist \size(b)$.
Thus, we put $(b, b')$ into BSPD.

\item[Case 2: $\dist(b, b'') > \udist \size(b) $:]
Here, from $\size(b) \le \size(\father(b'))$,  we have
 $\dist(b, \father(b')) \le \varrho \max(\size(b),
\size(\father(b'))) = \varrho
\size(\father(b'))$ (If this is not true, clearly in the WSPD
computation tree $T$, we will use $(b, \father(b'))$ instead of
$(b,b')$).
Similarly, we have $\dist(\father(b), b') \le
 \varrho \max(\size(\father(b)), \size(b'))$.

We now show how to find equal-sized floating-virtual
 boxes $b_2$ inside
the tight-virtual box  $\father(b)$ and  $b'_2$ inside the
tight-virtual box  $\father(b')$. This will identify the potential
edge-boxes.
Let $\Delta$ be the size of the equal-sized boxes $b_2$ and $b_2'$.
Then we have
\begin{eqnarray}
\begin{cases}
\dist(b_2, b_2') \ge \dist(b,b'') - 2 \sqrt {d} \Delta \\
\dist(b_2, b_2') \le \dist(b,b'')
\end{cases}
\end{eqnarray}
Then the following is clearly a
 sufficient condition for the existence of such a pair of
floating-virtual boxes $b_2$ and $b_2'$ to define the
potential-edge-boxes $b$ and $b'$:
\begin{eqnarray}
\begin{cases}
\dist(b_2, b_2') \ge \dist(b,b'') - 2 \sqrt {d} \Delta  \ge \ldist
\Delta \\
\dist(b_2, b_2') \le \dist(b,b'')  \le \udist \Delta
\end{cases}
\end{eqnarray}
Notice that $\size(b) \le \frac{\dist(b,b'')}{\udist}$.
Thus, it is equivalent to require $\Delta$ in range
\[
 \frac{\dist(b,b'')}{\udist} \le \Delta \le
\frac{\dist(b,b'')}{\ldist + 2\sqrt{d}}
\]
Clearly, we have a solution for $\Delta$ when
\begin{equation}
\label{eqn:condition-BSPD-1}
\udist  \ge \ldist + 2 \sqrt{d}
\end{equation}

We now show that the boxes $b_2$ and $b_2'$ will be inside the boxes
$\tvbox(\father(b)_h$ and $\tvbox(\father(b')_h$ respectively indeed.
To ensure this, we only need the condition that size $\Delta$ is at
 most $\min(\size(\father(b))/2, \size(\father(b'))/2 )$.

If $\size(\father(b)) \ge \size(\father(b'))$,  then we have
\begin{eqnarray}
\dist(b,b'') < \dist( b, \father(b'))  + 2 {\sqrt
  d}\size(\father(b'))
 \le (\varrho + 2 {\sqrt  d}) \size(\father(b'))  =
  (2\varrho + 4 {\sqrt  d}) \min( \size(\father(b))/2, \size(\father(b'))/2 )
\end{eqnarray}

If $\size(\father(b)) \le \size(\father(b'))$,  then we have
\begin{eqnarray}
\dist(b,b'') < \dist(\father(b), b'')  + 2 {\sqrt
  d}\size(\father(b))
 \le (\varrho + 2 {\sqrt  d}) \size(\father(b))  <
  (2\varrho + 4 {\sqrt  d}) \min( \size(\father(b))/2, \size(\father(b'))/2 )
\end{eqnarray}

Thus, when
\begin{equation}
\label{eqn:condition-BSPD-2}
2\varrho + 4 {\sqrt  d} \le \udist
\end{equation}
 we can choose $\Delta = \frac{\dist(b,b'')}{\udist} \le \min(
 \size(\father(b))/2, \size(\father(b'))/2 )$ as
 the final size of $b_2$ and $b_2'$ to ensure that the
 floating-virtual boxes $b_2$ and
 $b_2'$ will be  inside the boxes
 $\tvbox(\father(b))_h$ and $\tvbox(\father(b')_h$ respectively indeed.
Observe that we do not put any condition on the locations of $b_2$ and
 $b_2'$ here.
Thus, as in Lemma~\ref{lemm:finding-almost-equal-sized-box}
we can find arbitrary
locations of $b_2$ and $b_2'$ in time $O(d)$.
\end{enumerate}
Then a sufficient condition for the constants $\ldist$ and $\udist$
such that we can compute a BSPD  from the
 WSPD with a constant $\varrho$ is
\begin{equation}
\begin{cases}
\ldist \le \varrho - \sqrt{d}\\
\udist \ge 2\varrho + 4 {\sqrt  d}
\end{cases}
\end{equation}
It is easy to show that we need $\udist \ge 2 \ldist + 6 {\sqrt d}$.
This finishes the proof.
\end{proof}

Based on the above proof, we can find the floating-virtual boxes for
 each pair of boxes in the bounded-separated pair decomposition in
 time $O(d)$.
Consequently, we also can find  $\Nbr_{\ge}(b)$ for all boxes $b$ in
 linear-time.

\subsection{Cones and Cone Partition}
\label{subsec:cones}

We consider points in the $d$-dimensional  space $\dspace$.
Let $B= \{ z_1 , z_2 \ldots z_g \} $ be a set of
 $g$ linearly independent vectors.
The set of vectors $B$ is  called a {\it basis} of $\dspace$ for a
cone partition.
We define the \emph{cone} of $B$, $\Cone(B)$, as
\[ \Cone(B) =  \{ \sum_{1 \leq i \leq g} \lambda_i z_i \mid
 \forall i, \lambda_i \geq 0 \} \]
For a point $x$, we define the cone with apex at $x$ generated by
 a basis  $B$ as
$\RCone(x,B)= \{ y \mid   y-x \in \Cone(B) \}$.
For a given set of points $P \in \dspace$, the cone region
 of $P$ is defined as
\[ \RCone(P,B) = \bigcup_{x \in P}  \RCone(x,B).\]

A vector $uv$ is said to be  \emph{in the direction} of the cone
 $\Cone(B)$ (or cone with basis $B$) if $v \in \RCone(u,B)$.
A vector $uv$ is said to be  \emph{in the direction} of a family of cones
 with a collection of bases ${\cal B}=\{B_1, B_2, \cdots, B_m\}$ if
 $v \in \cup_{i=1}^{m}\RCone(u,B_i)$.
For two vectors $x$ and $y$, the angle between them is denoted as
 $\Ang(x,y) = \arccos (\frac{x\cdot y}{||x||\cdot ||y||})$, where $\|x\|$ is
 the length of the vector.
We define the {\it angular span}, based at the origin,
 of a set of vectors $B$ (and its corresponding cone) as
 $\Ang(B) = \max_{x,y \in B} \{ \Ang(x,y) \} $.
Given a set of bases $\cal B$, its angular span is defined
 similarly.
Let ${\cal F}$ be a finite family of basis of $\dspace$.
$\cal F$ is called a frame if $\bigcup_{B \in {\cal F}} \Cone(B) = \dspace$.
The angular span of a frame $\cal F$ is defined as
$\Ang( {\cal F} ) = \max_{B \in {\cal F}}  \Ang (B) $.
The following lemma has been shown in \cite{Yao82}.
\begin{lemma}
For any  $0 < \phi < \pi $, one can construct a frame $\cal F$ in
$\dspace$
with size $O((\frac{c}{\phi})^d)$ for a constant $c$
 such that $\Ang( {\cal F})  <  \phi $.
\end{lemma}

The \emph{Yao} graph, based on the cone partition of the space produced by
a frame $\cal F$, contains all edges $uv$, where $v$ is the closest
node to $u$ in some cone $\RCone(u,B)$ for $B \in {\cal F}$.
We use $Y(V, {\cal F})$ to denote such graph.
The following lemma was obvious.
\begin{lemma}
\label{lemm:yao}
The graph $Y(V, {\cal F})$ is a $t$-spanner with $O(n)$ edges
 for $t=\frac{1}{1-2\sin (\phi({\cal F}) / 2) }$
 when $\phi({\cal F}) < \pi /3$.
\end{lemma}

It is easy to construct an example of points
(for example, $n/2$ points placed evenly on a side of a unit square
 and another $n/2$ points placed evenly on the opposite side) such that
 the weight of  $Y(V, {\cal F})$ is
 $\Theta(n) \weight(\EMST)$.
The maximum node degree in $Y(V, {\cal F})$ could also be as large as
 $\Theta(n)$, \eg, when $n-1$ points evenly placed in a circle and $1$
 node at the circle center.

\paragraph{General-Cone-Direction Property:}
We now define the  {\em General-Cone-Direction Property} for a pair
 of boxes $b$ and $b'$.
Here we require that the pair of boxes $b$ and $b'$ have similar sizes,
 \ie, $\lsize \cdot \size(b) \le \size(b') \le \usize \cdot \size(b)$
 for fixed constants $\lsize$, $\usize=1/\lsize \ge 1$.
In addition, the distance $\dist(b,b')$ between the boxes $b$ and $b'$
 is at least $\ldist \max(\size(b), \size(b'))$ for a constant
 $\ldist >1$, where $\size(b)$ is the size of a box $b$.
We also assumed that the angular span of each cone in $\cal F$ is at most
 a constant $\alpha$ (depending on the spanning ratio $t$).

\begin{definition}[General-Cone-Direction Property]
Given a pair of boxes $b$ and $b'$,
 a general-cone-direction of $b'$ with respect to the base box $b$,
 denoted as ${\cal B}(b,b')$,
 is a set of cones ${\cal B}\subset {\cal F}$ such that, for every
 point $x \in b$,
 $\bigcup_{B_i \in {\cal B}}\RCone(x, B_i)$ properly contains the box
 $b'$.
Similarly, we can define the general-cone-direction ${\cal B}'(b,b')$ of $b$
 with  respect to $b'$.
Also, define $\theta$ as the maximum  angular span of ${\cal B}$ and
 ${\cal B}'$,
i.e. $\theta = \max \{ \Phi({\cal B}) , \Phi({\cal B}') \}$.
\end{definition}

A vector $uv$ is said to be  \emph{in the general-cone-direction} of
 $b'$ with respect to the box $b$
 if $uv$ is in the direction of ${\cal B}(b,b')$.

\begin{lemma}
\label{lemma:general-cone-direction-angle}
For any  pair of  bounded-separated boxes $b$ and $b'$ in a
Bounded-Separated decomposition,
the maximum angular span of ${\cal B}(b,b')$ and ${\cal B}'(b,b')$
is  at most  $\theta =  \frac{4 \sqrt  d}{\ldist} +
3 \alpha$.
\end{lemma}
\begin{proof}
Figure \ref{fig:small-angle} illustrates our proof that follows.
\begin{figure} [htpb]
\begin{center}
\scalebox{0.5}{\input{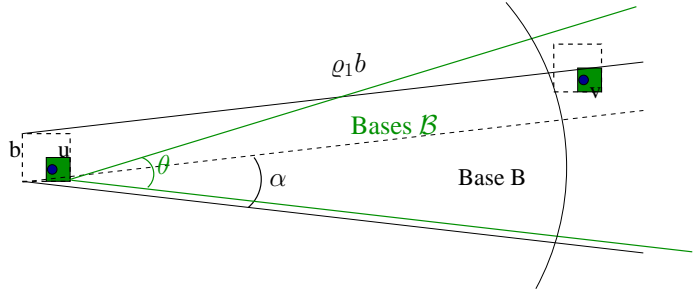}}
\end{center}
\caption{An illustration of the proof that we can choose $\alpha$ and
  $\ldist$  such that $\theta$ is small.}
\label{fig:small-angle}
\end{figure}

A simple computation shows that the general-cone-direction
 property is satisfied
 if we select a minimal collection $\cal B$ of bases from
 $\cal F$ such that all vectors in  $B_i \in {\cal B}$ are
 within angle $\theta_1$ from vectors in
 $B$ where
\begin{equation}
\label{cond:parameter-cover-all}
\theta_1 \le  \frac{2 \sqrt  d}{\ldist} + \alpha
\end{equation}
Then $\theta \le 2 \theta_1 +  \alpha \le \frac{4 \sqrt  d}{\ldist} +
3 \alpha$.
\end{proof}

Clearly, there is only at most a constant number of cones in $\cal B$
 and we can find $\cal B$ in constant time.

\subsection{Method for Low Weight $t$-Spanner Using Split Trees}

Given a set $V$ of $n$ nodes in $\dspace$ and any real number $t>1$,
 we now describe our  method to construct a structure $G = (V,E)$
 in time $O(n \log n)$
 such that $G$ is (1) $t$-spanner,
(2) each node in $V$ has a degree $O(1)$, and (3) the total edge
length $\weight(G)$ is $O(\weight(\MST))$, where $\MST$ is the
Euclidean minimum spanning tree over $V$.

The basic idea of our method is as follows.
We first construct the compressed split-tree  $CT(V)$ and associate
  geographic information with $CT(V)$. This includes
 enclosing-boxes,  the tight-virtual boxes, and the floating-virtual boxes.
By Theorem~\ref{theo:potential-edge-boxes}
 we can construct a bounded-separated pair decomposition (BSPD) in
 linear time using $CT(V)$.
For  each box $b$ in $CT(V)$, we will select one node as its
 \emph{representative} node, denoted as $\Rep(b)$,
  used as a gateway node
 to connect nodes inside this box to nodes at some other boxes in
 $\Nbr(b)$.
To ensure that a node is used as a representative node by at most a
 constant number of boxes (thus ensuring $O(1)$ degrees for these
 nodes), we apply the following strategy in
 selecting the representative nodes of  boxes.
Each leaf vertex will have at least one  node from $V$
 inside.
There are at most $n-1$ internal vertices in the compressed split-tree,
 and $2n-1$ vertices in $CT(V)$; thus we need at most $ 2 n$ representative
 nodes.
We will assign $2$ credits to  each node in $V$.
For each leaf box $b$ of $CT(V)$, we choose a node $\Rep(b)$ inside $b$
 and charge $1$ credit to the chosen node.
Using a  bottom-up approach, for each internal vertex $b$, we will
 select a node $\Rep(b)$  from nodes contained inside  $b$
 that  has a non-zero credit.
Since each internal vertex has at least $2$ children vertices, such
 a representative node can always be found.
Thus, we have the following lemma.

\begin{lemma}
\label{lemm:representative-times}
Each node $v_i$ in $V$ is used at most $2$ times as a
 representative node in $CT(V)$.
\end{lemma}

Our algorithm works as follows.
Given a pair of potential-edge boxes $b_1$ and $b_1'$ (defined by a
 pair of bounded-separated floating-virtual boxes $b_2$ and $b_2'$),
 and their representative nodes $u$ and $v$ respectively,
 we  add an edge $uv$ if:
\begin{compactenum}
\item
 there is no  edge $xy$ already added,
 where $x$ is inside $b_2$ (it is possible that $y \not \in b_2'$),
 $xy$ crossing the boundary of $b_2$ (thus
 $y$ is not inside $b_2$),
 and $xy$ is in the general-cone-direction
  of the basis ${\cal B}(b_2, b_2')$.
Let $B$ be the basis such that the representative node $v$
 is contained inside the region $\RCone(b_2, B)$.
Such an edge $xy$ is called \emph{crossing edge} for box $b_2$ in the
 direction of $B$, and

\item
 there is no crossing edge $zw$ already  added,
 where $z$ is from $b_2'$, $zw$ crosses the boundary of $b_2'$, and $zw$ is
 in the general-cone-direction of the basis ${\cal B}'(b_2', b_2)$.
\end{compactenum}


 Algorithm \ref{alg:kfts-k1} presents our method
 for constructing a $t$-spanner in $\dspace$ with low-weight, and
 bounded degree property.
In Algorithm \ref{alg:kfts-k1},
 for each enclosing-box $b$,  each basis $B_i \in {\cal F}$, and each
 dimension,
 we store an edge $xy$
 to array $CrossingEdge(b,B_i,h)$ (if there is any)
 such that (1) node $x$ is inside the enclosing-box $b$,  (2) node $y$ is in the
 cone $\RCone(x,B_i)$, and (3) $y$ is the node that is \emph{furthest} from the
 box $b$ in the dimension $h$
 if there are multiple edges satisfying the first two conditions.
This  will ensure the following lemma:
\begin{lemma}
For every direction specified by the basis $B_i$,
 there exists an edge $w_1 w_2$ with $w_1 \in b$ and
 $w_1w_2$ crossing a floating-virtual  box $b_2$ (at some ancestor of
 $b$) in the direction $B_i$
 if and only if there is an edge $xy \in CrossingEdge(b,B_i,h)$
 for a dimension $h$ and $xy$ crosses the virtual box $b_2$ (\ie $y
 \not \in b_2$).
\end{lemma}

\begin{algorithm}[tb]
\caption{Constructing a $t$-spanner with low-weight}
\label{alg:kfts-k1}
\begin{algorithmic}[1]
\STATE
Define a frame $\cal F$, with a constant $c$ number of bases
 $B_1$, $B_2$, $\cdots$, $B_{c}$ such that the angular span of any
 base is at most a small angle $\alpha$.
The actual value of the angle $\alpha$ will be given later in proofs.

\STATE
Build the compressed split-tree $CT(V)$ and a BSPD.
With each enclosing-box $b$ in the split-tree $CT(V)$,
 we  associate a representative node $\Rep(b)$.
For each enclosing-box $b$, we also construct $\Nbr_{\ge}(b)$.

For each enclosing-box $b$,  each basis $B_i \in {\cal F}$, and each
 dimension $h$,
 we define an array $CrossingEdge(b,B_i,h)$.


\STATE

Sort the edge-distances (see Definition \ref{def:edge-distance})
 between all pairs of potential-edge boxes in increasing order.
(There are a total  of $O(n)$ pairs of potential-edge
boxes, thus, the sorting can be done in time $O(n \log n)$.)
Ties are broken by the actual Euclidean distance between the
 representative nodes of the boxes.

\FOR{($r=1$ to $\sum_{b} \Nbr_{\ge}(b)$)}
\STATE
Select the pair of potential-edge  boxes $b_1$ and $b_1'$
 with the $r$th smallest edge-distance.
Let $b_2$ be the floating-virtual  box containing $b_1$ and $b_2'$ be the
 floating-virtual  box containing $b_1'$ such that $b_2$ and $b_2'$ are a
 pair of bounded-separated floating-virtual boxes.
Let $u$ (and $v$ resp.) be the representative node of box $b'_1$ (and
 $b_1$ resp.).

Let ${\cal B}, {\cal B}' \subset {\cal F}$  be the collection of bases
 satisfying the {\em General-Cone-Direction Property}
w.r.t $b_1$ and $b_1'$, respectively.
We then add an edge $uv$,  only if
\begin{compactenum}
 \item  $\forall B_i \in {\cal
 B}$, $\forall h$, such that $x$ is inside the box $b_1$ and
  $y$ is outside of the floating-virtual  box $b_2$,
 there is no ``crossing''
 edge $xy$ in $CrossingEdge(b_1,B_i,h)$; and

 \item
 $\forall B'_i \in {\cal B}'$, $\forall h$,
 such that $z$ is inside the box $b_1'$
 and $w$ is outside of the floating-virtual  box $b_2'$,
 there is no ``crossing'' edge $zw$ in $CrossingEdge(b'_1,B_i,h)$.
\end{compactenum}

After  adding edge $uv$, $\forall B_i \in {\cal B}$, $\forall B'_i \in
 {\cal B}''$ and   $\forall h \le d$,
 we update the array
 $CrossingEdge(b_1,B_i,h)$ and the array
 $CrossingEdge(b_1',B'_i,h)$ correspondingly (adding an edge $uv$ to
 the array if the edge $uv$ is added to the spanner).

\ENDFOR

\STATE Let $G=(V,E)$ be the graph constructed.
\end{algorithmic}
\end{algorithm}


\section{Properties: Low-Weight, Spanner, and Low-Degree} \label{sec:proofs-k1}

We next show that the constructed structure  $G$ by Algorithm
 \ref{alg:kfts-k1},  is a
 $t$-spanner, has a bounded degree, and has low-weight
 (by choosing the angular span of the frame $\cal F$, $\theta$, and
 the parameters $\ldist$, $\udist$, 
 $\lsize$, and $\usize$ in bounded-separateness carefully).

\subsection{Degree Property}

\begin{theorem}
\label{them:bound-degree-1}
Each node in the constructed graph $G$ by Algorithm \ref{alg:kfts-k1},
 $v \in V$, has degree
 $ \le |{\cal F}|=O((\frac{1}{\alpha})^d)$
 where $\alpha$ is the angular span of the frame $\cal F$.
\end{theorem}
\begin{proof}
Since each node will serve as a representative for at most two
 different enclosing-boxes,
 it suffices to show that for each enclosing-box $b_1$, we will add at most
 a constant number of  edges for the representative node of this box.
We will show that we add at most $1$ edge to a node $u$ in any cone
 direction when $u$ is a representative node of a box $b_1$.
Assume that we have already added an edge $uv$ in a direction $B$,
 where $u$ is the representative
 node in box $b_1$ and $v$ is the representative node of $b_1'$ such
 that $b_1$ and $b_1'$ is a pair of potential-edge boxes defined by a
 pair of bounded-separated boxes $b_2$ (containing $b_1$) and $b_2'$
 (containing $b_1'$).
We show that we cannot add another edge $uw$ in the same direction $B$
 later.
Assume that we did add another edge $uw$ later, because of the
 existence of a pair of
 bounded-separated virtual boxes $b_3$ and $b_3'$ that defines a pair of
 potential-edge boxes $b_1$ and $b'$ (for $b' \not = b_1'$).
Then there are only two complementary cases:
\begin{compactenum}
\item
$b_3$  contains boxes $b_1$ and $b_1'$ inside.
This violates the condition (condition 3) of the potential-edge
definition: $b_3$
 will contain the parent box (which is an enclosing-box) of $b_1$ and $b_1'$
 inside.
Notice that since a virtual box $b_3$ contains both $b_1$ and $b_1'$
 inside, our compressed split-tree construction shows that the parent box
 of the boxes $b_1$ and $b_1'$ is inside $b_3$ (may be same as $b_3$).
\item
$b_3$ does not contain $b_1'$ inside.
Then edge $uv$ will be a crossing edge that crosses the
 boundary of virtual box $b_3$. Thus edge $uw$ will not be added.
\end{compactenum}
This finishes the proof.
\end{proof}

Note that we later will show that the angular span $\alpha$ of $\cal
F$ depends on the spanning ratio $t>1$ that is  required. 

\subsection{The Spanner Property}

We now prove that the final structure $G$ is a $t$-spanner, where $t >1$
 is a given constant, if we
 choose $\theta$, $\ldist$ and $\udist$ carefully.

\begin{theorem}
\label{theo:spanner-1}
The final structure $G$ constructed by Algorithm \ref{alg:kfts-k1}
 is a $t$-spanner for a given constant $t>1$ if we carefully
 choose $\alpha=\Ang( {\cal F} )$, $\ldist$ and $\udist$ according to
 conditions 
 illustrated in (\ref{eqn:all-conditions-spanner}).
\begin{equation}
\label{eqn:all-conditions-spanner}
\begin{cases}
\ldist   >  t \sqrt d \\
2    {\sqrt d}  \frac{t+1}{t-1} \le  \ldist\\
\frac{2t \sqrt  d}{\ldist}+ (1+ \frac{2t \sqrt  d}{\ldist})
(1+\frac{\sqrt d}{\ldist}) / (1-2\sin (\theta/2) - 2\frac{{\sqrt d}}{\ldist})
  \le t \\
\theta \le  \frac{4 \sqrt  d}{\ldist} + 3 \alpha\\
\udist \ge 2 \ldist + 6 {\sqrt d}
\end{cases}
\end{equation}
\end{theorem}
\begin{proof}
Note that the last condition is to ensure that we can construct a BSPD.
We then prove the theorem by induction on the rank of the Euclidean
 distance between  all pairs of nodes $u$ and $v$ from $V$.
First, for the pair of nodes $u$ and $v$ with the smallest distance,
 edge $uv$ clearly will be added to $G$.
Thus, we have a path in $G$ with length at most $t \dist(u,v)$ to
 connect $u$ and $v$.
Assume that the statement is true for all pairs of nodes with the
  first $r$ smallest pairwise distance.
Consider a pair of nodes $u$ and $v$ with $(r+1)$th smallest distance.

Since we produce a BSPD for the set $V$ of nodes using $CT(V)$,
 in the box tree $CT(V)$, there will be a pair of floating-virtual boxes $b_2$
(containing  $u$) and $b_2'$ (containing $v$) that is a
bounded-separated pair.
Let $b_1$ be the largest  enclosing-box  (from tree $CT(V)$) that is
 contained inside  $b_2$ and contains $u$; and
 $b_1'$ be the largest  enclosing-box  (from tree $CT(V)$) that is
 contained inside  $b_2'$ and contains $v$.
Then the pair of boxes $b_1$ and $b_1'$ is a pair of potential-edge boxes.
Depending on whether we have a crossing edge  $xy$ when processing the
pair of potential-edge-boxes
 $b_1$ and $b_1'$, we have the following two complementary cases.

\textbf{Case 1}:
 We have an edge $xy$ where $x$ is a representative node of box $b_1$
and $y$ is a representative node of box $b_1'$.
In this case, we have $\dist(u,v) \ge \dist(b_1, b_1') \ge \dist(b_2,
b_2')\ge \ldist \max(\size(b_2), \size(b_2'))$.
By choosing
\begin{equation}
\label{eqn:conditions-spanner-1}
\ldist   >  t \sqrt d
\end{equation}
we have $\dist(u,x) \le \size(b_1) \cdot {\sqrt d} \le
\dist(b_2, b'_2)  {\sqrt d} / \ldist =
\boxdist(b_1, b'_1)  {\sqrt d} / \ldist
 \le \dist(u,v)  {\sqrt d} / \ldist < \dist(u,v) /t$.
Then by induction, we have a path connecting $x$ and $u$ with length
 at most  $ t \dist(u,x)$: this is true because
  this path can only use  edges
 with length  smaller than $\dist(u,v)$, and $\dist(u,x) < \dist(u,v)$.
Similarly, we have a path in $G$ connecting $v$ and $y$ with length
 at most  $ t \dist(v,y)$.
Thus, in the final structure $G$, we have a path (with subpath from
$u$ to $x$, subpath from $y$ to $v$, and edge $xy$) connecting $u$ and
 $v$ with length at most
\[
t \dist(u,x) +  t \dist(v,y) + \dist(x,y)
\le  (t+1) \dist(u,x) +  (t+1) \dist(v,y) + \dist(u,v)
 \le ( (2 (t+1)   {\sqrt d} / \ldist )+1)
  \dist(u,v)
\]
This is at most $t \dist(u,v)$ if
\begin{equation}
\label{eqn:conditions-spanner-2}
2    {\sqrt d}  \frac{t+1}{t-1} \le  \ldist
\end{equation}

\begin{figure} [htpb]
\begin{center}
\scalebox{0.4}{\input{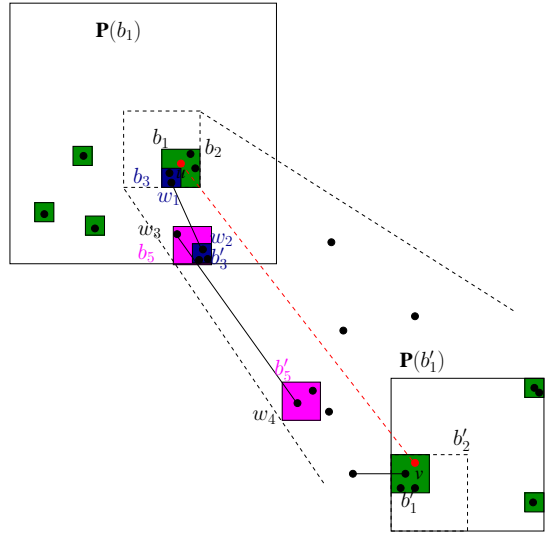}}
\end{center}
\caption{An illustration of the proof that the final structure $G$ is a
 $t$-spanner.
Here for a pair of nodes $u$ and $v$ we will have a path with length
  at most $t \dist(u,v)$, where $u$ and $v$ are
  representative nodes of the potential-edge boxes $b_1$ and $b_1'$.}
\label{fig:spanner-property}
\end{figure}

\textbf{Case 2}:
 We do not have an edge $xy$ where $x$ is a representative node of
  enclosing-box $b_1$
 and $y$ is a representative node of  box $b_1'$.
In this case,   one or both of the following conditions  is true:
\begin{compactenum}

\item
there is a crossing edge $w_1w_2$
 such that $w_1$ is inside $b_1$,  $w_2$
  is outside of floating-virtual box $b_2$ and $w_2$ is in the
  general-cone-direction of $\RCone(w_1, {\cal  B})$, or

\item
there is a crossing edge $w_1w_2$
 such that $w_2$ is inside $b_2$,  $w_1$
  is outside of floating-virtual box $b_1$ and $w_1$ is in the
  general-cone-direction of $\RCone(w_2, {\cal  B}')$.

\end{compactenum}
W.l.o.g., we assume that the first condition is true.
See Figure \ref{fig:spanner-property} for the illustration of the
 proof that follows.

Consider the general-cone-direction  ${\cal B}(b_1,b_1')$ of
 box $b_1'$ with respect to the base box $b_1$.
The set $\cal B$  of cones will be called a \emph{meta-cone}.
The angular span of $\cal B$ is at most a value $\theta$ (from Lemma
 \ref{lemma:general-cone-direction-angle}).
Observe that since the meta-cone $\Cone(w_1, {\cal B})$ will contain the box
 $b_1'$, it will also contain the node  $v$ inside.
Recall that the edge $w_1w_2$ has the same direction as the meta-cone
 $\cal B$,  the meta-cone $\Cone(w_1, {\cal B})$ also contains $w_2$ inside.

Then for the node $w_2$ and node $v$, they must be contained in a pair
 of boxes in BSPD from the definition of BSPD.
Consider the bounded-separated pair
 of floating-virtual boxes (say $s$ and $s'$) containing them respectively.
When the angle $\theta < \pi /3$, we have $\dist(w_2, v) < \dist(w_1,
 v)$.
Together with the fact that $w_2$ is outside of the floating-virtual box $b_2$,
 we can show that the edge-distance (\ie, $\dist(s,s')$)
 of the pair of potential-edge
 boxes containing $w_2$ and $v$ respectively is less than the edge-distance
 $\boxdist(b_1,b_1')$ between $b_1$ and $b_1'$.
In other words, the pair of nodes $w_2$ and $v$ has been processed
 before the pair of nodes $u$ and $v$.
Thus, we either will have a directed edge $w_3w_4$ such that $w_3$ and $w_4$
 are representative nodes of the boxes $s$ and $s'$ respectively;
 or we will have an edge $z_3z_4$ such that $z_3z_4$ has the same
 direction as the meta-cone ${\cal B}(s,s')$, \ie,
  $z_3z_4$ is inside the meta-cone ${\cal  B}(s,s')$.
Observe that the distance between the boxes  $s$ and $s'$ is
 smaller than the distance between the boxes  $b_1$ and $b_1'$.

We can repeat the above process and get a sequence of edges
 $w_1w_2$, $w_3w_4$, $w_5w_6, \cdots, w_{2k-1}w_{2k}$, by renaming the
 nodes and the pairs of potential-edge boxes, and the pairs of
 bounded-separated floating-virtual boxes, with the following properties:
\begin{compactenum}
\item
$w_1$ is inside an enclosing-box $b_1$ and $w_2$ is outside of floating-virtual
  box $b_2$ containing $b_1$ (if it
  is not, we can pick the first one $w_{2i-1}w_{2i}$ such that this
  property is satisfied);
node $w_2$ is inside an enclosing-box, called $b'_1$, which is inside a
  floating-virtual box, called $b_2'$.
The pair of boxes $b_2$, $b_2'$ is a pair of bounded-separated
  floating-virtual boxes.
Observe that here the boxes $b_1$, $b_1'$, $b_2$ and $b'_2$
 may be different from what we called at the beginning of our proof.
\item
In general, for $i  \ge 1$,
 node $w_{2i-1}$ is inside an enclosing-box $b_{2i-1}$ which is
 inside a floating-virtual box $b_{2i}$;
 node $w_{2i}$ is inside an enclosing-box $b'_{2i-1}$ which is
 inside a floating-virtual box $b'_{2i}$.
Here, for $i \ge 1$, the pair of bounded-separated floating-virtual boxes
 $b_{2i}$, $b_{2i}'$ contain the pair of potential-edge boxes
 $b_{2i-1}$, $b_{2i-1}'$, which is
  used to define the edge $w_{2i-1}w_{2i}$,
 \ie, $w_{2i-1}$ (resp. $w_{2i}$) is a representative node of
 the enclosing-box $b_{2i-1}$ (resp. $b'_{2i-1}$).
Notice that here either the box $b_{2i-1}'$ or the box $b_{2i+1}$
 could be the larger one between them, although both contain node
 $w_{2i+1}$.
 We also have that the node $w_{2i+1}$ is inside the enclosing-box $b'_{2i-1}$
 for $i \ge 1$, while
 $w_{2i}$ is outside of the floating-virtual box $b'_{2i}$, for $i \ge 1$.
\item
The angle $\angle v w_{2i-1}w_{2i} \le \theta$ for a value $\theta$ in
Lemma \ref{lemma:general-cone-direction-angle}.
\end{compactenum}

Thus, we have a path
\[
u\leftrightsquigarrow w_1 w_2 \leftrightsquigarrow w_3w_4 \cdots
  w_{2i-1}w_{2i} \leftrightsquigarrow  w_{2i+1}w_{2i+2}
\leftrightsquigarrow \cdots
 \leftrightsquigarrow w_{2k-1}w_{2k} \leftrightsquigarrow v
\]
to connect the pair of nodes $u$ and $v$.
Here $p \leftrightsquigarrow q $ denotes a path constructed
recursively to connect nodes $p$ and $q$.
By induction, we know that the length of path $u\leftrightsquigarrow
w_1$ is at most
 $t \dist(u,w_1) \le t \sqrt{d} \size(b_1) \le  t \sqrt{d}
\frac{\dist(u,v)}{\ldist} $; similarly
 the length of the path $w_{2i}\leftrightsquigarrow  w_{2i+1}$ is at
 most $ \max( t \sqrt{d} \size(b_{2i-1}'), t \sqrt{d}
 \size(b_{2i+1}))$
 since either (1) $w_{2i}$ and $w_{2i+1}$ are inside $b_{2i-1}'$ or
 (2)  $w_{2i}$ and $w_{2i+1}$ are inside $b_{2i+1}$.

Notice that $\max(\size(b_{2i-1}'), \size(b_{2i-1}))
\le \frac{\dist(w_{2i-1}, w_{2i} )}{\ldist} $ from the definition of
 potential-edge boxes.
Additionally, $\max(\size(b_{2i-1}'), \size(b_{2i-1})) \le \usize
\min(\size(b_{2i-1}'), \size(b_{2i-1}))$ since the floating-virtual boxes
 $b_{2i-1}'$ and $b_{2i-1}$ are required to have similar sizes (within
 a factor $\lsize=1/\usize$ of each other).
Then the total length of the path
 $w_1 w_2 \leftrightsquigarrow w_3w_4 \cdots w_{2i-2}
\leftrightsquigarrow  w_{2i-1}w_{2i} \leftrightsquigarrow \cdots
 \leftrightsquigarrow w_{2k-1}w_{2k}$ is at most
\[(\sum_{i=1}^{k} \dist(w_{2i-1}, w_{2i})) \cdot (1+ \frac{2t \sqrt d}{\ldist})
\]
Thus, the length of the path $u\leftrightsquigarrow w_1 w_2
\leftrightsquigarrow w_3w_4 \cdots w_{2i-2} 
\leftrightsquigarrow  w_{2i-1}w_{2i} \leftrightsquigarrow \cdots
 \leftrightsquigarrow w_{2k-1}w_{2k} \leftrightsquigarrow v$
 is at most
\begin{equation}
(\sum_{i=1}^{k} \dist(w_{2i-1}, w_{2i})) \cdot (1+ \frac{2t \sqrt
  d}{\ldist})
+  \frac{2t \sqrt  d}{\ldist} \cdot \dist(u,v)
\end{equation}
We then bound the length $\sum_{i=1}^{k} \dist(w_{2i-1}, w_{2i})$.
From the general-cone-direction property, when $\theta  < \pi /3$,
 it is easy to show that
\[\dist(v, w_{2i-1}) - \dist(v, w_{2i}) \ge (1-2\sin(\theta/2))
 \dist(w_{2i-1}, w_{2i}).\]
Since $ \dist(v, w_{2i+1}) - \dist(v, w_{2i})   \le \dist(w_{2i},
 w_{2i+1}) \le {\sqrt d}\max(\size(b_{2i-1}), \size(b_{2i+1}) )
 \le \frac{{\sqrt d}}{\ldist} \max(\dist(w_{2i-1}, w_{2i}),
 \dist(w_{2i+1}, w_{2i+2}) )$,
we have
\[  \dist(v, w_{2i}) - \dist(v, w_{2i+1})  \ge - \frac{{\sqrt
 d}}{\ldist} \max(\dist(w_{2i-1}, w_{2i}),
 \dist(w_{2i+1}, w_{2i+2}) ).\]
Then, we have
\begin{eqnarray*}
\dist(u,v) +\dist(w_1,u) & \ge & \dist(w_1, v) \ge
  \sum_{i=1}^{k}[\dist(v, w_{2i-1}) - \dist(v,
  w_{2i})] +  \sum_{i=1}^{k-1}[\dist(v, w_{2i+1}) - \dist(v,
  w_{2i})] \\
& \ge & (1-2\sin (\theta/2) - 2\frac{{\sqrt d}}{\ldist})
  \sum_{i=1}^{k} \dist(w_{2i-1}, w_{2i})
\end{eqnarray*}
Consequently, the ratio of the length of the path we found over
$\dist(u,v)$ is at most
\begin{equation}
\label{cond:parameter-spanner}
\frac{2t \sqrt  d}{\ldist}+ (1+ \frac{2t \sqrt  d}{\ldist})
(1+\frac{\sqrt d}{\ldist})
/ (1-2\sin (\theta/2) - 2\frac{{\sqrt d}}{\ldist})
  \le t
\end{equation}
when $\theta$, and $\ldist$ are chosen carefully ($\theta$ is small
enough and $\ldist$ is large enough).

This finishes the proof of the spanner property.
\end{proof}

It is easy to show that we can  carefully
 choose $\alpha=\Ang( {\cal F} )$, $\ldist$ and $\udist$ that satisfy the
 conditions   in (\ref{eqn:all-conditions-spanner}).
Notice that these conditions are weaker than
 the conditions required to achieve low weight property,
 illustrated in (\ref{eqn:all-conditions}).

\subsection{The Weight Property}
\label{subsec:weight-1}

We next show that the weight $\weight(G)$ of the graph $G$ constructed is
 $O(\weight(MST))$.
Our proof technique is based on  the proofs
used in \cite{das1993oss,das-soda95}.
Recall that an edge $e$ is added to graph $G$ when we process a pair of
  potential-edge boxes that are defined by a pair of bounded-separated
 floating-virtual boxes $b$ and $b' \in \Nbr(b)$.
We then say that floating-virtual boxes $b$, $b'$ and the edge $e$ form a
 \emph{dumbbell} (as defined in \cite{das1993oss}).
For a dumbbell formed by edge $e=(u,v)$, for both node $u$ and node $v$, we
 associate a cylinder with each node,
 and call it \emph{dumbbell head}, of suitable size.
A dumbbell head is a cylinder of radius $ \delta_1 \|uv\|$ and height
 $\delta_2 \| uv\|$ with $0 < \delta_1 \ll \delta_2 \ll 1$.
These dumbbell heads are always contained inside the corresponding
 floating-virtual boxes.
Similar to  \cite{das1993oss}, we can group edges of $G$ into
 $g=O(1)$ groups $E_1, E_2, \cdots, E_g$  such
 that for edges in each group $E_i$, we have
\begin{compactenum}
\item[\textbf{Near-Parallel Property:}] any pair of edges $u_1u_2$ and $v_1v_2$
  in a group are nearly  parallel, \ie, the angle formed by vectors
  $u_2-u_1$ and  $v_2-v_1$ is bounded by a constant $\theta_0$.
 This clearly can
  be achieved using a partition based on cones: we first use a constant number
 of cones to  partition the space $\rspace^d$ (where the angular span
  of the cone base is at most $\theta_0$).
Then each cone defines a  group of edges: all the (directed) edges $uw$
 contained in the direction of this cone.

\item[\textbf{Length-Grouping Property:}] In a group, any two edges have
  lengths that are either nearly equal or differ by more than a large
  constant factor.
This can be achieved by first grouping edges into buckets (the $i$th
  bucket contains edges with lengths in $[\delta^{i+1} L, \delta^{i}L ]$
  where $L$ is the length of the longest edge and constant $\delta \in
  (0,1)$).
Then form a group as the  union of  every
 $s$th bucket (so the edge lengths from different buckets
  differ by at least $\delta^{s}$ factor).

\item[\textbf{Empty-Region Property:}] In a group, any two edges that
 have nearly
  equal length of value $x$ are far apart, \ie, the distance between
 end-nodes of these edges are  at   least $\epsilon_1 x$ for some
 constant $\epsilon_1$.
Here $\epsilon_1 >0$  could be any constant (even larger than $1$).
This clearly can be done by showing that for each
  edge $uv$ of length $x$, there are at most $O(1)$ edges that are of similar
  length and are not far apart (that has at least one
 end-node  within distance  $\epsilon_1 x$ of an end-node of $uv$).
Recall that, in our method,
  for every added edge $uv$, we will only add at most 1 edge for
  the pair of potential-edge boxes defining $uv$ and the size of the
  floating-virtual  boxes $b_2$ and $b_2'$
  is at least a constant fraction of the edge length $\dist(u,v)$.
The virtual boxes $b_2$, $b_2'$ used to add an edge $uv$ will be used to
  define the dumbbells of the edge $uv$.
Recall that
  the  virtual boxes  will  be either disjoint or one is completely
  contained inside the other.
This implies that, given any edge $uv$,
 there is only a constant number of edges $xy$ that
  are of similar length and are nearby edge $uv$.

\ignore{ 
\item[\textbf{Nested-Dumbbell Property:}] For any two dumbbells associated with
  two edges  in a group $E_i$, either they are completely disjoint, or
  one dumbbell is
  contained inside the dumbbell head of the other larger dumbbell.
\begin{figure}[hptb]
  \begin{center}
      \epsfxsize=2.5in\epsfbox{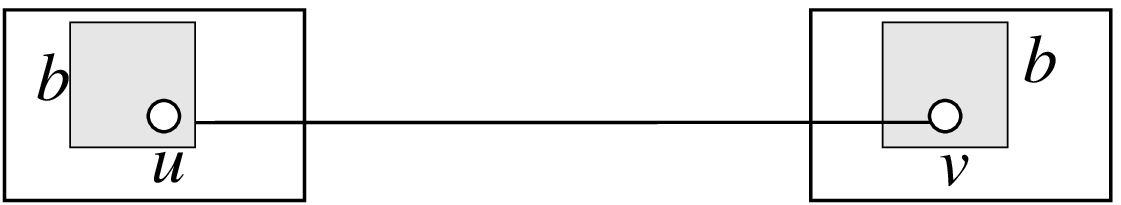}
    \caption{Illustration of the dumbbell of an edge. The shaded
    squares are boxes processed when we add an edge $uv$ and the
    cylinders (rectangles in 2D) are dumbbell heads.}
 \label{fig:example-dumbbell}
  \end{center}
\end{figure}
For an edge $uv$, we have two boxes $b_u$ (containing $u$) and $b_v$
(containing $v$) of the same size $\size(b_u)$ and distance
$\dist(b_u, b_v) \in [\ldist \size(b_u), \udist \size(b_u)]$.
Notice that, the edges in one group are almost parallel and the
 sizes of any two  edges
 are either  similar  (if so then their dumbbells are far apart)
 or differ by at least $\delta ^s$ for $\delta \in (0,1)$ and some
 integer constant $s$.
Assume that we have a sequence of edges $e_1=(u,v)$, $e_2$, $\cdots$, $e_m$
 such that a box , that created,
 edge $e_i$ contains an end-point of edge $e_{i+1}$.
Assume that the box $b_u$ of node $u$ contains an end-point of $e_2$.
Then the distance of the end points of edges $e_2$, $e_3$,  $\cdots$,
 $e_m$ to $u$ is at most $\size(b_u) + \|uv\| \sum_{i=1}^{m-1}
\delta^{s i} \le \|uv\| (\frac{1}{\ldist} + \frac{\delta^{s}}{1-\delta^s})$.
Since the edges are almost parallel, it is also easy to show that the
distance of end-points of these edges $e_2,e_3,\cdots, e_m$ to the
line passing through $e_1$ is at most $\sum_{i=2}^{m}
\frac{\dist(e_i)}{\ldist} \le \frac{\delta^s}{\ldist
  (1-\delta^{s})} \|uv\|$.
Then similar to \cite{das1993oss}, we can ensure the nested-dumbbell
 property by choosing a proper cylinder height $\delta_2 \|uv\|$ and cylinder
 radius $\delta_1 \|uv\|$ with $0 < \delta_1 \ll \delta_2 \ll 1$.
The details are omitted due to space limit.
The final resulted dumbbells are not perfect cylinders, but their
shapes are  close to being a cylinder.

\item[\textbf{Isolated-Centers Property:}] In every group $E_i$, for any two
  edges $e_1$ and $e_2$ with different sizes (ratio of lengths is
  larger than $\delta^s$), the end-node of the longer edge is \emph{not}
  inside  a dumbbell head associated with the smaller edge.
This also can be done similar to the empty-region property by showing
  that the \emph{dumbbell conflict graph} has constant degree for a
  group of edges that already satisfies length-grouping property and
  nested-dumbbell property.
Here dumbbell conflict graph $CG$ are defined over  dumbbells (associated
  with edges) and two dumbbells are connected by an edge in $CG$ if
  the isolated-centers property is violated.
}
\end{compactenum}

Consequently, we have the following lemma:
\begin{lemma}
We can group edges of $G$ into $O(1)$ groups such that the edges in each group
 satisfy the preceding properties: near-parallel, length-grouping, and
 empty-region.
\end{lemma}

Here the number of groups produced depending on the values $\ldist$,
$\udist$, $\lsize$, $\usize$, and $\aratio$.
Recall that the bounded aspect ratio is at most $\aratio \le 2$ for all the
\emph{tight-virtual} boxes.
However, the aforementioned properties do not ensure that the total edge
weight of edges in a group is $O(1)
 \weight(SMT_i)$ where $SMT_i$ is the Steiner minimum tree connecting
 the endpoints of edges in $E_i$.
We can construct an example of edges  \footnote{Place $n$ nodes evenly
  on a line and connect every pair of nodes. Then there is a group of
  edges produced by the preceding partitioning will have a total edge
  weights of $O(\log n) \weight(SMT_i)$.}
 satisfying the aforementioned properties
 such that the total edge weights
 could be as large as $O(\log n) \weight(SMT_i)$.
To prove that the graph produced by our method is low-weighted, we
 need an additional property:

\medskip

\begin{compactenum}
\item[\textbf{Empty-Cylinder Property:}]
 for every edge $uv$ and its associated
 dumbbells, there is a cylinder (with the  height $\ge \eta_1
 \dist(u,v)$ and  radius  at least $ \eta_2\dist(u,v)$
 for some positive constants $\eta_1$ and $\eta_2$)
 with axis using some segment of the edge $uv$ such that the cylinder is
 empty of any end-node of edges in the same group.
This cylinder is called a \emph{protection cylinder} of the edge $uv$.
\end{compactenum}

Observe that the empty-region property does not imply the
 empty-cylinder property, and neither does the empty-cylinder property implies
 the empty-region property.

\begin{lemma}
\label{lemma:empty-cylinder}
By carefully choosing $\ldist$,  and $\alpha$ (and thus $\theta$),
 according to conditions illustrated in inequality
 (\ref{cond:parameter-empty-cylindar})
 every added edge $uv$ by our Algorithm \ref{alg:kfts-k1}
 has the empty-cylinder property.
\end{lemma}
\begin{proof}
Assume that $uv$ is added due to the pair of potential-edge boxes
$b_1$ and $b_1'$, which is defined by a pair of
bounded-separated (floating-virtual) boxes $b_2$ and $b_2'$.
Thus $b_1$ and $b_1'$ are  contained inside $b_2$ and $b_2'$ respectively.
Let $B$ be the  base such that $v \in \RCone(b_2, B)$
 and let $\cal B$ be the minimal collection of bases such that
 for any point $p$ inside the box $b_2$, $b_2' \in \RCone(p, {\cal B})$, i.e.,
 bases that are in the {\em general-cone-direction}  $\cal B$.

Since $uv$ is added, we know that there is no edge $xy$ crossing $b_2$
 with $x \in b_2$ and $xy$ is in the general-cone-direction $\cal B$.
We will show by a simple contradiction  that there is a node $p$, such
 that the cone $\RCone(p, {\cal B})$ is empty of nodes $w \not \in
 b_2$ with  distance $\dist(p,w) \le \dist(p, b_1')$.
If this is not true,
 consider all the pairs of nodes $x$ and $y$ with $x \in b_2$, $y
 \not \in b_2$, $y \in \RCone(x, {\cal B})$ and $\dist(x,y) \le
 \dist(x, b_1')$.
Let $p,q$  be the pair with the smallest distance among all such pairs
 of nodes $x,y$.
Then edge $pq$ will exist in the graph $G$, which contradicts the
 existence of edge $uv$.

Since the cone $\RCone(p, {\cal B})$ is empty of nodes,
 then by choosing a large enough $\ldist$, we will have a large
 empty-cylinder at the middle of the segment $uv$.
For example, if we let $\ldist$ be four times of the
 value of $\ldist$ that satisfies condition
 (\ref{cond:parameter-cover-all}), \ie,
\begin{equation}
\label{cond:parameter-empty-cylindar}
\theta \le 3\alpha + \frac{16 \sqrt  d}{\ldist}
\end{equation}
 then we have an empty-cylinder near
 the middle of the segment $uv$
 with height almost half of the length $\dist(u,v)$.
In other words, if  condition (\ref{cond:parameter-empty-cylindar}) is
 satisfied, we have $\eta_1 \simeq  1/2$,
 and $\eta_2 =\ldist/(4\udist)$.
Recall that here $\ldist$ and
 $\udist= \Theta( \ldist)$ are constants used to define the
 bounded-separateness of two almost-equal-sized virtual boxes.
\end{proof}

Thus, for any edge $uv$ added by our method, we know that there is a
 cylinder using a segment $wz$ of $uv$ as axis with radius at
 least $\eta_2\|uv\|$ for a constant $\eta_2$,  $wz$ has length
 $\|uv\|/2$ and in the center of segment $uv$.

\begin{definition}[Isolation Property]\cite{das-soda95}
A set of edges $E$ is said to satisfy the \emph{isolation property} if
\begin{compactenum}
\item
 With every edge $e=uv \in E$ can be associated a cylinder $C(e)$
 whose axis is a segment of $uv$, and the size of the cylinder is not
 small,  \ie, the height is at least $\eta_1 \dist(u,v)$
 and the radius of the basis is at least $\eta_2  \dist(u,v)$ for
 some  positive  constants $\eta_1$ and $\eta_2$.
\item
For every edge $e$,  its associated cylinder $C(e)$ is not
 intersected by any other edge.
\end{compactenum}
\end{definition}

The following theorem was proved in \cite{das-soda95} by Das \etal.
\begin{theorem}\cite{das-soda95}
\label{theo:isolation}
If a set of edges $E$ in $\dspace$ satisfies the isolation property,
 then $\weight(E)=O(1) \weight(SMT)$, where $SMT$ is the
 Steiner minimum spanning tree connecting the endpoints of $E$.
\end{theorem}

Based on this theorem, we then show that the graph $G$ produced by our
 method is also low-weighted.
Observe that a group of edges from the graph $G$, partitioned as
 previously to satisfy the near-parallel, length-grouping, and
 empty-region properties, may not satisfy the isolation
 property directly.

\begin{theorem}
\label{theo:isolation-empty-region}
The set of edges $E_i$ that satisfies \textbf{empty-region property}
 and \textbf{empty-cylinder property} has
 a total weight at most $O(\weight(SMT_i))$ where $SMT_i$ is the
 Steiner minimum spanning tree that spans the vertices in $E_i$.
\end{theorem}
\begin{proof}
We first use the grouping approach to partition the edges into a
constant number of groups $E_i$, $1 \le i \le g$,
 with each group of edges satisfying the near-parallel, length-grouping, and
 empty-region properties.
It now suffices to study the weight of a group $E_i$.
We essentially will show that, for each group $E_i$ of edges produced,
 we can remove some edges such that
 (1) the total length of all removed edges is bounded by a constant
 factor of the total length of the remaining edges, and
 (2) the set of the remaining edges satisfies the isolation property.
If these two statements were proven to be true, the theorem then
 directly follows.
Figure \ref{fig:low-weight} illustrates the proof that will follow.
\begin{figure*}[hptb]
 \begin{center}
\begin{tabular}{cc}
 \epsfxsize=2.5in\epsfbox{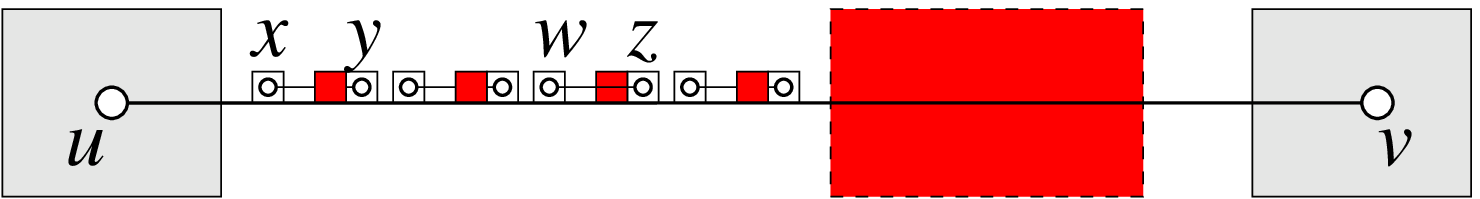} &
\epsfxsize=2.5in\epsfbox{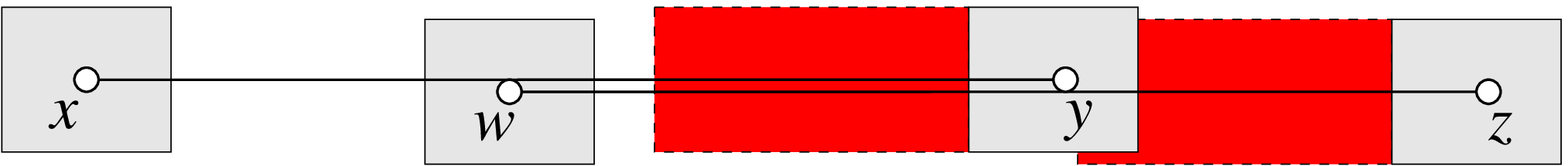}  \\
(a) & (b)
\end{tabular}
  \end{center}
    \caption{Illustration of the proof of the low weight property.
Here (a) a long edge may intersect the cylinders (dark shaded
 rectangles) of many shorter edges
 and  (b) the length of the similar sized edges that are intersected by
 an edge is bounded.}

\label{fig:low-weight}
\end{figure*}

Recall that when we added an edge $uv$ to the graph $G$, edge $uv$
 satisfies the Empty-Cylinder Property.
It is easy to show that the cylinder associated with $uv$
 will not intersect any edge $xy$
 (with a much shorter length) added before $uv$ and  $xy \in E_i$.
This can be done by shrinking the protection cylinder by at most a
 small constant factor.

On the other hand, it is possible that edge $uv$ may intersect  cylinders of
 many edges $xy \in E_i$ with shorter lengths.
See Figure \ref{fig:low-weight} (a) for an illustration of such case.
Here the black shaded regions are protection cylinders.
Given an edge $uv$, let $I(uv)$ be the set of edges $xy \in E_i$ such
 that $uv$ intersects the protection cylinder of the edge $xy$.
We then process edges $uv \in E_i$, starting from the longest edge, to
 produce $E_i'$ as follows: the longest edge $uv \in E_i$ is added to
 $E_i'$ and update $E_i$ by removing all edges $I(uv)$ from $E_i$; we
 repeat this procedure till $E_i$ is empty.
Clearly, the set of edges $E_i'$ satisfies the \emph{isolation
 property} and thus has total weight at most $\weight(SMT)
 =O(\weight(EMST))$.
Observe that, the preceding processing of $E_i$ is only for the proof
 of the low-weight property; we will not remove these edges $I(uv)$
 from the constructed structure $G$.
Also observe that all edges $xy \in I(uv)$, where $xy$ has a length at
 most $\delta^{s}\dist(u,v)$, $s > 1$, must be inside a cylinder with axis
 $uv$, height almost $\dist(uv)$, and radius at most $\eta_2
 \delta^{s} \cdot \dist(u,v)$.

We then show that the length of edges in $I(uv)$ is at most
$O(\|uv\|)$.
For simplicity, we assume that all edges in $I(uv)$ are parallel to
 $uv$.
The rest of the proof will still hold (with different constants)
 since the edges in $E_i$ are almost parallel.
Recall that the edges in $I(uv)$ have the length grouping property and
 the length  of every edge is at most $\|uv\|$.
Let $I_i(uv)$ be  edges from $I(uv)$
  with length in the range of $[\delta^{s\cdot i+1} \dist(u,v),
  \delta^{s\cdot i} \dist(u,v)]$, where $i\ge 0$ is any integer.
Here $s \gg 1$ and $0<\delta \ll 1$ are positive constants used in
deriving the length grouping property.
Let $X_i$ be the total length of edges in $I_i(uv)$.
We first show that $X_1$ is at most
 $\epsilon \dist(u,v)$ for a small constant $0< \epsilon <1$.
Let $x_1y_1$, $x_2y_2$, $\cdots$, $x_ay_a$ be edges from $I_1(uv)$
 such that the projection $x_i'$ of $x_i$ on edge $uv$
 is at the righthand side of the projection $y_{i-1}'$
 of node $y_{i-1}$ on edge $uv$.
Then $\sum_{j=1}^{a} \dist(x_jy_j) \le (1-\eta_1) \dist(u,v) + 2
\delta^s \dist(u,v) \le \epsilon \dist(u,v)$
 for a constant $\epsilon=1-\eta_1+2\delta^s <1$ when integer
 $s$ is chosen large enough.
Here $\eta_1$ is the constant used to define the ratio of the
 height of  a protection cylinder over the length of the edge.

We then show that we will \emph{not} have edges $wz$ in $I_1(uv)$
 such that it will have
  endpoints $w$ such that its projection $w'$ on $uv$
 is in the interval $[x_j', y_j']$ for some $1 \le j \le a$.
Figure \ref{fig:low-weight} (b) illustrates the situation for this case.
This can be proved by choosing a large constant $\epsilon_1 \gg 1$ in defining
 the  empty-region property: for any edge $x_jy_j$, there is no edge
 $wz$ of similar length such that $w$ is within distance $\epsilon_1
 \dist(x_j,y_j)$ of $x_j$ or $y_j$.

We can show that the protection cylinders defined
 by edges in $I_1(uv)$ are disjoint, and these protection cylinders
 are also disjoint from the protection cylinder of the edge $uv$.
Obviously, any edge $wz$ from $I_i(uv)$ cannot
 have node $w$ or $z$ falling inside the protection cylinders of
 the edges in $I_j(uv)$ for $j < i$.
Recall that our choices of  protection cylinders (their sizes)
 already  ensure that any edge $wz$ from $I_i(uv)$ cannot intersect
 the cylinders of edges from  $I_{t}(uv)$ with $t \le i$.
Thus, the total length of edges in $I_{i}(uv)$, denoted as $X_i$, is
 at most \footnote{With a small additive value whose total length over all $i$
 is bounded by $X$. This is for the case that we may have an edge $xy$
 from $I_i(uv)$ such that $x$ is outside of the cylinder using $uv$ as
 axis and radius proportional to $\delta^{si} \dist(u,v)$, and $y$ is
 inside this cylinder. Obviously, the total length of such edges are
 at most $X=\dist(u,v)$.}
\[X_i \le X - \eta_1 \sum_{t=0}^{i-1} X_t.\]
Here $X=\dist(u,v)$, and
 $\eta_1 \sum_{t=1}^{i-1} X_t$ is the total height of the
 protection cylinders defined by edges in $I_0(uv)=\{uv\}$,
 $I_1(uv)$, $I_2(uv)$, $\cdots$, $I_{i-1}(uv)$.
These protection cylinders are empty of nodes,
  and also empty of edges from $I_i(uv)$.
Then it is easy to show by induction that, for any $i \ge 1$,
 $\sum_{t=0}^{i} X_{t} \le \eta X$ for a constant
 $\eta= \frac{1}{\eta_1}$.
This finishes the proof.
\end{proof}

Thus, by choosing the parameters $\alpha$, $\ldist$, $\udist$,
 and $\theta$ satisfying the following conditions
\begin{eqnarray}
\label{eqn:all-conditions-all}
\begin{cases}
\ldist \ge 4 \sqrt{d} & {\mbox{from Theorem
  \ref{theo:potential-edge-boxes}}} \\
\udist \ge 6\sqrt{d} +2 \ldist & {\mbox{from Theorem
  \ref{theo:potential-edge-boxes}}} \\
\theta \ge 3 \alpha + \frac{16 \sqrt d}{\ldist} & {\mbox{from Lemma
  \ref{lemma:empty-cylinder}}}\\
\ldist   >  t \sqrt d & \mbox{from Theorem \ref{theo:spanner-1}} \\
2    {\sqrt d}  \frac{t+1}{t-1} \le  \ldist & \mbox{from Theorem
  \ref{theo:spanner-1}} \\
\frac{2t \sqrt  d}{\ldist}+ (1+ \frac{2t \sqrt  d}{\ldist})
(1+\frac{\sqrt d}{\ldist}) / (1-2\sin (\theta/2) - 2\frac{{\sqrt d}}{\ldist})
  \le t & \mbox{from Theorem \ref{theo:spanner-1}} 
\end{cases}
\end{eqnarray}
 our structure is a $t$-spanner, with bounded degree, and is
 low-weighted.
To satisfy the aforementioned conditions, it suffices to satisfy
 the following  conditions when $t\le 3$
\begin{eqnarray}
\label{eqn:all-conditions}
\begin{cases}
\udist \ge 6\sqrt{d} +2 \ldist & {\mbox{from Theorem
  \ref{theo:potential-edge-boxes}}} \\
\theta \ge 3 \alpha + \frac{16 \sqrt d}{\ldist} & {\mbox{from Lemma
  \ref{lemma:empty-cylinder}}}\\
  \ldist \ge 2    {\sqrt d}  \frac{t+1}{t-1}  & \mbox{from Theorem
  \ref{theo:spanner-1}} \\
\frac{2t \sqrt  d}{\ldist}+ (1+ \frac{2t \sqrt  d}{\ldist})
(1+\frac{\sqrt d}{\ldist}) / (1-2\sin (\theta/2) - 2\frac{{\sqrt d}}{\ldist})
  \le t & \mbox{from Theorem \ref{theo:spanner-1}} 
\end{cases}
\end{eqnarray}

It is easy to show
 that we do have solutions for $\alpha$, $\ldist$, $\udist$ and
 $\theta$.
For example,  if we let $x = \sqrt d /\ldist$, $\alpha=x$,
 and $\theta=19 x$.
By choosing $21x <1$, we get
 $x = \frac{25t +1 - \sqrt{(25t+1)^2 - 160t(t-1)}}{80t}$ is a valid solution.
Thus, $\frac{t-1}{25t +1} \le x=\frac{\sqrt  d}{\ldist}  \le
 \frac{2(t-1)}{25t +1}$ clearly 
 is a solution satisfying the last condition here.
This solution also satisfies the other conditions in Inequalities
 (\ref{eqn:all-conditions}). 
Thus the number of cones produced is proportional to $(\frac{1}{\alpha})^d
 = O((\frac{t+1}{t-1})^d) = O((\frac{1}{t-1})^d)$.
Then we have our main theorem:
\begin{theorem}
\label{theo:main-k1}
Given a set of nodes $V$ in $\dspace$,
 Algorithm \ref{alg:kfts-k1} constructs a structure that
 has low-weight ($O(\weight(MST))$), 
has a constant bounded degree ($O((\frac{1}{c})^d)$),
and is a $t$-spanner
in time $O((\frac{1}{c})^d n \log n)$, where the constant $c$ depends
 on the spanning ratio $t$ and $\frac{1}{c^d}$ is the number of cones
 needed in our cone decomposition.
\end{theorem}

\section{$(k,t)$-VFTS Spanner Using Compressed Boxtree}
\label{sec:boxtree-k}

In this section, we show how to extend the previous approach to
 an efficient method for constructing a $k \ge 1$
 fault-tolerant $t$ spanner for any given set of nodes $V$ in
 $\rspace^d$, $t >1$ and $k \ge 1$.

\subsection{The Method}

Our approach follows the construction of $t$-spanners as in the previous
sections. We will assume the construction of
 the compressed split-tree $CT(V)$ and a bounded-separated pair
 decomposition (BSPD) based on $CT(V)$.
As has already been proven  previously (Lemma \ref{lemm:representative-times}),
 at each box of $CT(V)$, we can choose $1$ representative node so that
 each node $v_i \in V$ is used at most $2$ times as a representative
 node of some enclosing-boxes.
Since we want a $k$-VFTS, we will choose $k+1$ representative nodes
 for an enclosing-box, if it
 contains at least $k+1$ nodes inside.
For easy of presentation, we define various boxes.

\begin{definition}
We call an enclosing-box $b$ \emph{a $k$-box} if it contains at least
 $k+1$ nodes inside; otherwise it is called  \emph{a non-$k$-box}.
A $k$-box is called a \emph{leaf-$k$-box} if it is a $k$-box and none of its
children in the compressed split-tree is a $k$-box.
\end{definition}

For each $k$-box $b$, we will choose $k+1$ nodes contained inside $b$
 as the representative nodes $\Rep(b)$ of the box $b$.
We will discuss in detail on how to choose  $\Rep(b)$ for a $k$-box
 later.
For any box $b$ that contains at most $k$ nodes inside,
 all nodes will serve as the representative nodes $\Rep(b)$ of this box.

The basic idea of our method is as follows.
Consider a pair of potential-edge boxes $b_1$ and $b_1'$,
 and their $k+1$ representative nodes (if both boxes are $k$-boxes).
Here $b_1$ and $b_1'$ are tight-virtual boxes.
Consider the pair of
 bounded-separated floating-virtual  boxes $b_2$ and $b_2'$, with $b_2$
 contained inside the tight-virtual box $\father(b_1)$ and  $b_2'$
 contained inside the tight-virtual box $\father(b_1')$.
We  add  an edge
 $uv$ where $u$ is a representative node of $b_1$ and $v$
 is a representative node of $b_1'$, while the following is true:
\begin{compactenum}
\item
there are less than $k+1$ \emph{disjoint} edges of the form $xy$ where $x$ is
from $b_2$ and $xy$ is in the general-cone-direction  defined by $\RCone(x,{\cal B})$
such that  $\cal B$
 satisfies the {\em General-Cone-Direction property} w.r.t. $B$.
 Here $B$ is the cone basis such that the representative node $v$
 is contained inside the region $\RCone(b_2, B)$; \textbf{and}
\item
there are less than $k+1$  \emph{disjoint} edges of the form $zw$ where $z$ is
 from $b_2'$ and $zw$ is in a general-cone-direction ${\cal B}'$
 w.r.t. $B'$ where $B'$ is the basis such that the representative node $u$
 is contained inside the region $\RCone(b_2', B')$;
\end{compactenum}

\textbf{Data Structures Used:}
In our method, for each enclosing-box $b$, each dimension $h$,
 and each cone basis $B_i$ of a frame $\cal F$, we store a set of
 at most $k+1$  \emph{disjoint} edges $x_iy_i$
 in an array DisjointCrossingEdge, denoted as $\DisjointCrossingEdge(b,B_i)$,
 such that 
\begin{compactenum}
\item Node $x_i$ is inside the enclosing-box $b$,  
\item Node $y_i$ is in the  cone $\RCone(x_i,B_i)$.
\end{compactenum}

Another set AllCrossingEdge, denoted as $\AllCrossingEdge(b,B_i)$,
 stores all edges
 $xy$ with $x \in  b$, $y \not \in b$ and $y \in \RCone(x,B_i)$.
Our method will ensure that the cardinality
 of $\AllCrossingEdge(b,B_i)$ is at most $(k+1)^2$ for each cone
 direction $B_i$.
Using AllCrossingEdge array $\AllCrossingEdge(b,B_i)$
 for all cones $B_i \in {\cal B}$,
 we can find the maximum number of disjoint edges
 $\DisjointCrossingEdge(b_1,{\cal B})$ in the general-cone-direction 
 $\cal B$, which cross the box $b_1$ (one node inside $b_1$ and one
 node outside of $b_1$).
This can be done using a maximum matching in the bipartite graph over
 two sets of nodes: one set is all nodes inside $b$ and the other set
 is all nodes $y \not \in b$ with an edge $xy$ for some node $x \in b$.
Let $X(b_1)$ be the set of end-nodes of edges
 $\DisjointCrossingEdge(b_1,{\cal B})$ that is inside $b_1$.
Let $X(b'_1)$ be the set of end-nodes of edges
 $\DisjointCrossingEdge(b'_1,{\cal B}')$ that is inside $b'_1$.
Clearly, for $X(b_1)$ and $X(b'_1)$, each has size at most $k+1$.
For any node $u$ and a direction $B$, let $\deg(u,B)$ be the
 number of edges incident onto $u$ in the direction of $B$.
These edges were added because of processing some pairs of
  potential-edge boxes  with smaller distance.
At each step in the algorithm, edges are added based on the degree of
the representative nodes and edges in the array,
$\DisjointCrossingEdge$.
Our detailed method for constructing a $(k,t)$-VFTS
 is presented in Algorithm \ref{alg:kfts-k2}.

\begin{algorithm}[tbhp]
\caption{Constructing a $(k,t)$-VFTS Spanner with Low-weight}
\label{alg:kfts-k2}
\begin{algorithmic}[1]

\STATE
Build the compressed split-tree $CT(V)$.
For each enclosing $k$-box $b$ in the  tree $CT(V)$,
 associate it with $k+1$  representative nodes $\Rep(b)$.
For each box $b$,
 we sort the edge-distances $\boxdist(b,b')$ for $b' \in \Nbr_{\ge}(b)$.
We also sort all the distances
$\{\boxdist(b,b') \mid b' \in \Nbr_{\ge}(b), \text{for every box } b\}$
in time $O(n\log n)$.

We then process  pairs of potential-edge boxes $b$ and $b'$
 in increasing order of their distance $\boxdist(b,b')$.

Initiate the array $\AllCrossingEdge(b,B_i)$ and
$\DisjointCrossingEdge(b,B_i)$ as empty.

\FOR{($r=1$ to $ \sum_{b} \Nbr_{\ge}(b)$)}

\STATE
Consider an enclosing-box $b_1$ with $p$ representative nodes $u_1, u_2
  \ldots u_p, 1\leq p\leq k+1$,
 and the enclosing-box  $b_1' \in \Nbr(b)$ with $q$
 representative nodes $u_1', u_2' \ldots u_q', 1 \leq q\leq k+1$ such
 that the distance $\boxdist(b_1, b_1')$ has rank $r$ among all
 pairs of potential-edge boxes.
Let $b_2=\virtualbox(b_1)$ and $b_2'=\virtualbox(b'_1)$ be  a
 pair of bounded-separated boxes.

Let $B \in {\cal F}$
 be the  base such that the majority representative nodes of $b_1'$
 is contained inside $\RCone(b_2, B)$.
Let $B' \in {\cal F}$
 be the  base such that the majority representative nodes of $b_1$
 is contained inside $\RCone(b_1, B')$.
Let ${\cal B}, {\cal B}' \subset {\cal F}$  be the collection of bases
 satisfying the {\em General-Cone-Direction Property} w.r.t $B$ and
 $B'$ respectively.

\STATE

Let $l$ be the maximum of the cardinality of
  $\DisjointCrossingEdge(b_1,{\cal B})$ and
  $\DisjointCrossingEdge(b_1',{\cal B}')$.

\STATE
Let $Y_k(b, B)$ be the  sorted list of all nodes
$\{u\mid  u \text{ is inside box } b, u \not \in X(b),
 \text{ and } \deg(u,B) \le k \}$, in increasing order of $\deg(u,B)$.
Thus, we update $Y_k(b_1, B)$ and $Y_k(b_1', B')$.

\IF{($p \leq k$) and ($q \leq k$)}
\STATE
Add all edges $u_iu_j$ for all pairs of $u_i$ and $u_j$ inside $b_1$.
Add all edges $u_i'u_j'$ for all pairs of $u'_i$ and $u'_j$ inside $b'_1$.
\FOR{each node $u_i$ in $b_1$ with $\deg(u_i,B) \le k$}
\STATE
We add an edge $u_i u'_j$ to a node $u'_j$ in $b_1'$ where
 $u'_j$ has the smallest degree $\deg(u'_j,B')$ among all
nodes inside $b'_1$.
Update the degree for all nodes and arrays $\AllCrossingEdge$ and
$\DisjointCrossingEdge$.
\ENDFOR
\FOR{each node $u'_j$ in $b'_1$ with $\deg(u'_j,B') \le k$}
\STATE
We add an edge $u_i u'_j$ to a node $u_i$ in $b_1$ where
 $u_i$ has the smallest degree $\deg(u_i,B)$ among all
 nodes inside $b_1$ and $u_iu'_j$ was not added before.
Update the degree for all nodes and arrays $\AllCrossingEdge$ and
$\DisjointCrossingEdge$.
\ENDFOR
\ENDIF

\IF{($p \ge k+1$) and ($q \ge k+1$)}
\STATE
Add $k+1 -l $  edges of the form $u_iu'_i$, where
$l=\max(|\DisjointCrossingEdge(b_1, B)|, |\DisjointCrossingEdge(b'_1, B')|)$.
Here $u_i$, $i \le k+1-l$, are the first $k+1-l$ nodes in  $Y_k(b_1, B)$
 and $u'_i$, $i \le k+1-l$, are the first $k+1-l$ nodes in  $Y_k(b'_1, B')$.
\STATE
Update the set $\DisjointCrossingEdge(b,B_i)$  and
 $\AllCrossingEdge(b,B_i)$  accordingly.
\ENDIF

\IF{($p \leq k$) or ($q \le k$), but not both}
\STATE  Without loss of generality, we assume that $p \le k$ and $q
\ge k+1$.

\STATE
Add all edges $u_iu_j$ for all pairs of $u_i$ and $u_j$ inside $b_1$.
 Here $b_1$ contains exactly $p$ nodes inside.

\FOR{each representative $u_i$ of $b_1$}
\STATE Let  $g'$ be the cardinality of
 $\DisjointCrossingEdge(b_1',{\cal B}')$.

\STATE If $|Y_k(b_1', B')| \ge \min(k+1-\deg(u_i,B), k+1- g')$, we add
$\min(k+1-\deg(u_i,B), k+1- g')$ edges from $u_i$  to
 the first $k+1-\deg(u_i,B)$  nodes  in $Y_k(b_1',B')$;
 otherwise, we add   $|Y_k(b_1', B')|$ edges from $u_i$  to
  nodes  in $Y_k(b_1',B')$.
\STATE Update the array  $Y_k(b_1, B)$ and  $Y_k(b'_1, B')$, and the
set $\DisjointCrossingEdge(b_1',B'_i)$  and
 $\AllCrossingEdge(b_1',B'_i)$  accordingly.
\ENDFOR
\ENDIF

\ENDFOR

\STATE Let $G=(V,E)$ be the graph constructed.
\end{algorithmic}
\end{algorithm}

\subsection{Properties:  ($k,t$)-VFTS Spanner, Low-Weight, and Low-Degree}


\begin{lemma}
The set AllCrossingEdge $\AllCrossingEdge(b,B_i)$ has size at most
$(k+1)^2$ when $k \ge 1$.
\end{lemma}
\begin{proof}
Observe that, for each node $u$ in the box
$b$, we add at most $k+1$ edges in any direction $B_i$.
The moment the  set $\AllCrossingEdge(b,B_i) $ has size $(k+1)^2$,
 we must have at least $k+1$ disjoint edges crossing the box $b$ in
 the direction $B_i$.
This follows from the observation that we will not add any more edges
in our algorithm if  $\DisjointCrossingEdge(b,B_i)$ is at least $k+1$.
\end{proof}

We show that the constructed graph is $(k,t)$-VFTS.
Recall that a structure is called $(k,t)$-VFTS if for every pair of
nodes $u$ and $v$ either (1) edge $uv$ is presented or (2) there exist
$k+1$ node disjoint paths connecting $u$ and $v$ such that the length
of each path is at most $t \|uv\|$.

\begin{theorem}
\label{theo:spanner}
The graph $G$ is a $(k,t)$-VFTS.
\end{theorem}
\begin{proof}
We prove that for every pair of nodes $u$ and $v$, either edge $uv \in
G$ or there are $k+1$ internally vertex disjoint paths $\path_1$,
$\path_2$, $\cdots$, $\path_{k+1}$ connecting them and each path has
length at most $t \cdot \| uv \|$.
In other words, $G$ is a $(k,t)$-VFTS.
We will prove this by induction on the distance between
$b$ and $b'$, the bounded-separated floating-virtual boxes containing $u$ and
$v$, respectively. We will refer to this distance as the
distance between $u$ and $v$.
It is easy to show that the pair of nodes $v_1,v_2$ with the shortest
 distance  will be in $G$.

Assume that for any $1 \le j \le r$, for the pair of nodes whose
distance is the $j$-th smallest, our statement is true.
We then consider the pair of nodes $u$ and  $v$ with the $(r+1)$-th
smallest distance.

Consider the boxes $b$ and $b'$ such that $u \in b$ and $v \in b'$ and
$b$ and $b'$ are bounded-separated floating-virtual boxes and let $b' \in
\RCone(b,B)$ for some cone basis $B \in {\cal F}$.
Assume for simplicity that the direction  defining  $B$ is horizontal
and $b$ is aligned with the coordinate axis.
For notational simplicity, let $|b|$ be the number of nodes from $V$
that are inside the box $b$.

We have three complementary cases depending on the number of points in
$b$ and $b'$: (1) $|b|=p \leq k$, (2)  $|b'|=q \leq k$, and (3)
  $p > k+1$ and $q >k+1$.

First consider the case $|b|=p \leq k$.
The case $|b'|=q \leq k$ follows the same proof.
When we processed a box $b'$ from $\Nbr_{\ge}(b)$,  let $q=|b'|$.
If $q  \le k+1-g$ (where $g$ is the number of edges incident on $u$ in
 the direction of $B$),
 then edges are added between the node $u$ and  all nodes of
 $b'$.
Thus, we have an edge $uv$ already.
Otherwise we know that node $u$ is connected to
 $k+1$ nodes, say $w_1$, $w_2$, $\cdots$, $w_{k+1}$ with the condition
 that the distance between $u$ and $w_i$ is no more than $uv$ since we
 processed the boxes $b'$ in $\Nbr(b)$ in increasing order of
 distance to $b$.
Additionally, since we only focus on edges in one specific cone direction, we
 know the length of $w_i v$ is also less than that of $uv$.
Then the general proof below applies.

In general, w.l.o.g.,
 assume that neither $u$ nor $v$ are representative nodes in $b$  and $b'$.
In this case, we must have $p > k+1$ and $q >k+1$.
Let the $k+1$ representative nodes for $b$ and $b'$ be
 $U = \{ u_1 \ldots u_{k+1} \}$
 and $ U' = \{ v_1, \ldots v_{k+1} \}$  respectively.
The following cases arise
\begin{enumerate}
\item[Case 1:]  
For all $ i \le k+1$ there are edges $u_iv_i$
  connecting $u_i \in U$ and $v_i \in U'$.
Note that since $b$ and $b'$ are bounded-separated, 
 $\dist(u,u_i) < \dist(u,v)$, for $u_i \in U$
 and $\dist(v,v_i) <\dist(u,v)$ for $v_i \in U'$.
Thus, by induction, there exist $k+1$ disjoint paths between $u$ and
$u_i \in U$,
and there exist $k+1$ disjoint paths between $v$ and $v_j \in U'$.
Since $\forall i$, $u_iv_i$ is part of $G$, by using Mengers theorem
it is easy to
see that there are $k+1$ vertex disjoint
paths between $u$ and $v$ in $U'$.
It is also easy to show that the length of each of such $k+1$ paths is
at most $t \dist(u,v)$.

\item[Case 2:] There is an $i$, such that $u_iv_i$ does not exist. If $l$
  edges of the form
$u_iv_i$ are added where $u_i$ and $v_i$ are representatives of
$b$ and $b'$ then one of the following subcases arises:

\begin{compactenum}
\item
There are $k+1-l$ edges of the form $(x,y)$ where $x \in b$ and
$y \in \RCone(x,{\cal B})$, $\cal B$ in the {\em general cone
  direction} as $B$.
Let the nodes inside $b$ be $x_1, x_2 \ldots x_m, m=k+1-l$ and
 the edges satisfying the preceding condition be $x_iy_i$.
Here $y_i$ is a node inside some other box in the direction of $B$
 of the box $b$.
Note that distance between $y_i, 1\leq i \leq m$ and nodes $ v_j$ (by
 measure of the distance $\boxdist(b,b')$ for two edge-boxes $b$
 containing $y_i$ and $b'$ containing $v_j$)
 is less than the distance between $u$ and $v_j$.
We can thus apply induction to show that  there are $k+1$ vertex disjoint paths
between $u$ and $x_j, \forall j,  1 \leq j \leq m$. And by induction there are
$k+1$ node disjoint paths between $y_i$ and $v_j, \forall j,i$.
The result follows from Mengers theorem that there are $k+1$ node
 disjoint paths connecting the $k+1$ representatives $U$ to $k+1$
 representatives $U'$.
\item
A similar result is true when there are edges $zy$, $y \in b'$ and
$z \in \RCone(z, {\cal B}')$, where
 ${\cal B}'$ in the same general direction as $B'$,
 where $b \in \RCone(b', B')$.
\end{compactenum}
\end{enumerate}
This finishes our proof that the structure is $k$-fault tolerant.
Similar to the proof of Theorem \ref{theo:spanner-1}, we can show
 that  for every pair of nodes $u$ and $v$, each of the $k+1$
 disjoint paths found in the preceding constructive proof has a length
 at most $t \| u-v\|$.
Thus,  the structure we constructed is a $(k,t)$-VFTS.
\end{proof}

\begin{theorem}
\label{theo:low-weight-1}
In the graph $G=(V,E)$ constructed by our method, $\weight(E) = O(k^2)
\cdot \weight(\MST)$ when $k \ge 1$.
\end{theorem}
\begin{proof}
Consider the  edges added at every node in $CT(V)$.
We group edges into $O(k^2)$ groups and will show that the total length of
edges in each group is at most $O(\weight(\MST))$.

Note that for a pair of potential-edge boxes $b$ and $b'$,
 it is possible that $O(k^2)$ edges are added during
 the construction procedure.
This happens when there are $p < k+1$ (with $p = \Theta(k)$)
 representatives in one  box during the procedure.
As in Section \ref{sec:proofs-k1}
 the edges added to connect  representative nodes of
 pairs of potential-edge boxes can be partitioned into $O(k^2)$
 groups such that the edges in each group satisfy the properties
 outlined:  \textbf{the near-parallel property,
 length-grouping property,  empty-region property, and empty-cylinder
 property}.

To partition edges into groups, we assume that the representative
 nodes $U$  inside an enclosing-box $b$ are numbered as $u_1$, $u_2$,
 $\cdots$, $u_p$,  where $p \le k+1$ and the representative
 nodes $U'$ inside an enclosing-box $b'$ are numbered as $v_1$, $v_2$,
 $\cdots$, $v_q$,  where $q \le k+1$.
We group edges sequentially in increasing order of the distance
 between
 the pair of the potential-edge boxes $b$ and $b'$.
For simplicity, we will only consider one cone base $B$ and all edges
 added  in the direction of $B$.
Since there are only a constant number of cone bases, if the edges
 added in the direction of any cone is at most $O(k^2)\weight(\MST)$,
 the total weight of all edges is still $O(k^2)\weight(\MST)$.
For each edge $uv$ in the cone direction of $B$, we will put it into
one of the $k^2$ groups: $E_{i,j}$, $1 \le  i, j \le k$.

Notice that, for a pair of potential-edge boxes $b$ and $b'$,
 we add edges based on rules defined for three different
 cases  based on $|b|=p$, $|b'|=q$:
Case 1) $|b|=p \leq k$, Case 2)  $|b'|=q
 \leq k$, and Case 3)   $p \ge k+1$ and $q \ge k+1$.

In case 1, $|b| \leq k$, for each $u_i$,
 we add $\min(k+1, q)$ edges $u_iv_j$ for some nodes $v_j \in U'$.
Then the edge $u_iv_j$ is added to group $E_{i,j}$ for $i \le k+1$ and $j
 \le k+1$.
Notice that we also added $p(p-1)/2 < k^2$ edges inside the box $b$.
Observe that each such added edge (inside the box $b$) has length at
 most a small constant fraction of the edge $u_iv_j$ (added to connect
 representative nodes of $b$ and $b'$).
We will \emph{not} add these edges to any group $E_{i,j}$.
A simple charging
 (charge the total edge length to one such crossing edge $u_iv_j$)
 method shows that these omitted edges have total edge length at most
 $O(k^2)$ times the total edge length of edges in a group $E_{i,j}$.
We  will  show that $\weight(E_{i,j})$ is at most $\weight(\MST)$.

For case 2, for each $v_i$ we will add $\min(k+1, p)$ edges $u_jv_i$
 for some nodes $u_j\in U$ and then
 we similarly add the edge $u_jv_i$ to group $E_{j,i}$.
Thus, all edges added in these cases are put into different groups,
 \ie, for any group $E_{i,j}$, and any pair of potential-edge boxes
 $b$ and $b'$, we have at most one edge $uv$ with $u$ from $b$ and $v$
 from $b'$.

The third case is that both boxes $b$ and $b'$ have at least $k+1$
nodes inside.
Let $b_2$ and $b_2'$ be the pair of bounded-separated boxes that contain
  the  boxes $b$ and $b'$ respectively.
In this case, we will add $k+1-l $ edges,  $u_iv_i$,
 where  $u_i \in U$ and $v_i \in V$.
Recall that we added these $l$ edges because
\begin{compactenum}
 \item there are at most $l \le k+1$ disjoint 
 edges already leaving the floating-virtual box $b_2$ in
the direction of the cone $B$, and
 \item there are at most $l \le k+1$
 disjoint edges already leaving the floating-virtual box
 $b_2'$ in  
 the direction of the cone $B'$.
\end{compactenum}
Notice that those  edges were added when processing a pair of
potential-edge boxes with shorter distance.
Thus they have already been put into some groups 
 (at most $2l$ different groups, since there are  at most $2l$ such edges).
Then a newly added edge  $u_iv_i$ is  put into a group  that is
 different from those $2l$ groups.
Notice that this is always possible since $2l + k+1 - l \le 2(k+1)
 \le (k+1)^2$ when $k \ge 1$.
Thus, when we put an edge $u_iv_i$ into some group, we know that
 there is no edge $xy$ in the same group such that
 $x$ is inside $b$ (resp. $b'$) and edge $xy$ crossing the boundary of
 floating-virtual box $b_2$ (resp. $b_2'$) in the direction of $\cal B$
 (resp. ${\cal B}'$).
Then similar to Theorem \ref{theo:isolation-empty-region}, we can
 prove that  each group $E_{i,j}$ of  edges can be further partitioned
 into a constant number of subgroups such that each subgroup of edges
 satisfies  all properties  outlined previously, and thus
 the total weight of all edges in each group  $E_{i,j}$ is at most
$O(\weight(\MST))$.
This finishes the proof of the theorem.
\end{proof}

\begin{theorem}
\label{theorem:bound-degree-k}
Each node $v \in V$ has a degree  $O(k)$ in the graph $G$ when $k \ge 1$.
\end{theorem}
\begin{proof}
Notice that for each node $u \in V $, we  add edges incident to
 it only when (1)  it is a node contained inside some non-$k$-box,
 or (2) when it is a representative node in some $k$-box.
For simplicity, we will only concentrate on edges added in the
 direction of one cone.

\begin{compactenum}
\item
We first study how many edges will be added to a node $u$ when $u$ is
 inside some non-$k$-box.
When node $u$ is inside a non-$k$-box, let $b_u$ be the largest
 non-$k$-box that contains $u$ inside.
Assume that $b_u$ contains $p < k+1$ nodes inside.
According to our method, we will add $p-1$ edges to other $p-1$ nodes
 inside the box $b_u$, and add
  $\min(k+1-g_u, q)$ edges for a potential-edge box $b'$ with $q$
 nodes inside, where $g_u$ is the number of edges incident on $u$, in
 the direction of $B$, and with shorter distances.
We have thus added at most $p-1 + k+1=O(k)$ edges in the direction of $B$
 when we have processed all enclosing-boxes $b'$ in $\Nbr(b)$, when $u$ is a node
 inside some non-$k$-boxes.

\item
We then study how many edges $uv$ will be added to $u$ when we
 process a pair of potential-edge boxes $b$ and $b'$ such that $b$ is
 a $k$-box, $b'$ is a non-$k$-box, such that  $u$ is
 inside the $k$-box $b$ and $v$ is inside the non-$k$-box $b'$.

Let $b_u^1=\father(b_u)$ be the smallest enclosing-box that contains $u$
 inside and is a $k$-box.
In other words, box $b_u^1$ is a leaf-$k$-box.
We adopt a charging argument where we assign credits to nodes
inside $k$-boxes to account for edges added to those nodes.
With each node inside the $k$-box is assigned
$(k+1)$ TYPE-1-credits for each cone direction
(another set  of TYPE-2-credits  is assigned in the next case).

Now let us see what will happen when we process the enclosing-box $b_1=b_u^1$.
The total free TYPE-1-credits of this box $b_u^1$ 
required to charge of edges in a certain cone direction is  $|b_u^1| (k+1)$ 
We will prove by induction that

\medskip

\begin{lemma}
\label{lemm:enough-credits}
 Every $k$-box $b$ will have  at least $(k+1)(k+1) - e_2 $ free TYPE-1-credits
 where $e_2$ is the number of edges that 

 (1) are added during the processing of a pair of potential-edge boxes $c$ and
 $c'$, where $c$ is a descendant $k$-box of $b$ and $c'$ is a
 non-$k$-box in $\Nbr(c)$, and

 (2) cross the boundary of $b$ in
 the given cone direction when we  process this box $b$  and its
 potential-edge boxes to add some edge after its children boxes have
 been processed.
\end{lemma}
\textbf{Proof of Lemma \ref{lemm:enough-credits}}:
This is clearly true for all leaf-$k$-boxes since it has been assigned
 $(k+1)$ TYPE-1-credits  for each node and a given direction, and it has at
 least $k+1$ nodes inside.
In the rest of the proofs, when we count the number of edges crossing
 the boundary of $b$, we will only count the edges that are added
 during the  processing of a pair of potential-edge boxes $b$ and $b'$ for some
 non-$k$-box $b'$.
The edges added when $b'$ is a $k$-box will be studied in the
 subsequent case.

Observe that when we process $b_1$  and all boxes
 $b_1' \in \Nbr(b_1)$ where $b_1'$ is a non-$k$-box,
 the total number of edges added to the node $u \in b_1$
 in any cone $B$ direction is at most $k+1$.
Consider  the case $b_1= b_u^1$.
Let $b_1'$, $b_2'$, $b_3'$, $\cdots$, $b_i'$, $b_{i+1}'$,
 $\cdots$, $b_{t}'$ be all potential-edge boxes that are non-$k$-boxes
 (with respect to
 $b_1$) in the direction of a given cone $B$, \ie, $b_j' \in
 \RCone(b_1, B)$ and $|b_j'| \le k$.
We further assume that $\dist(b_j', b) < \dist(b_{j+1}', b) $
 for $1 \le j \le t-1$.
Notice that some of these enclosing-boxes may be inside the box $\father(b_1)$,
 while some of them may be outside of $\father(b_1)$.
Assume that the first $i$ non-virtual potential-edge boxes
 $b_1'$, $b_2'$, $b_3'$, $\cdots$, $b_i'$ are inside
 $\father(b_1)$, while the rest of potential-edge boxes
  $b_{i+1}'$, $\cdots$, $b_{t}'$ are outside of the box
  $\father(b_1)$.
There are two cases here: $x_1=\sum_{j=1}^{i} |b_{j}'| \ge k+1$ and
 $x_1=\sum_{j=1}^{i} |b_{j}'| < k+1$.

\begin{compactenum}
\item
First consider the case $x_1 = \sum_{j=1}^{i} |b_{j}'| \ge k+1$.
Assume that $f \le i$ is the smallest index such that
 $\sum_{j=1}^{f} |b_{j}'| \ge k+1$.
We consider the number of total free TYPE-1-credits we will have for the enclosing-box
 $\father(b_1)$ in this direction.
The total TYPE-1-credits charged to nodes inside $b_1$
 for adding edges in this direction
 when processing enclosing-boxes  $b_1'$, $b_2'$, $b_3'$, $\cdots$, $b_i'$
 is at most $(k+1)(k+1)$ since we add at most $k+1$ edges for the
 ``closest'' $k+1$ nodes inside $b_j'$, $j \le f$
 and we will \emph{not} add any edges from
 nodes inside $b_j'$, $j >f$, to nodes inside $b_1$.
Observe that these new enclosing-boxes  $b_1'$, $b_2'$, $b_3'$, $\cdots$,
 $b_i'$ will also contribute at least $  (k+1) x_1$ TYPE-1-credits to
the enclosing-box  $\father(b_1)$.
Thus, the box $\father(b_1)$ has at least $x_1 \ge k+1 $ nodes, each
 with $  (k+1)$ free TYPE-1-credits left.
In other words, our claim holds.
Recall that a node gets a free TYPE-1-credit when it first becomes a node
 in some $k$-box and is thus not charged here.
Thus, any node $v$ inside boxes  $b_1'$, $b_2'$, $b_3'$, $\cdots$,
 $b_i'$ will \emph{not} be charged any of its TYPE-1-credits when we add
 edges $vu$ between  $b_j'$ and $b_1$, $j \le i$.

\item
We then consider the case that $x_1 = \sum_{j=1}^{i} |b_{j}'| <
k+1$.
Let $x_{2}= \sum_{j=1+i}^{t} |b_{j}'|$ be the number of nodes inside
non-$k$-boxes outside of $b_1$.
In this case, every node inside  $b_1'$, $b_2'$, $b_3'$, $\cdots$,
$b_i'$ will connect $k+1$ edges to some nodes inside $b_1$.
So the total TYPE-1-credits charged to nodes inside $b_1$ is $(k+1) x_1$.
The nodes inside $\cup_{j=1+i}^{t} b_{j}'$ will contribute at most
 $\min(x_2, k+1-x_1) \cdot (k+1)$ edges since only each of the closest
$\min(x_2,  k+1-x_1)$ nodes will  connect up to $k+1$ edges to nodes
 inside $b_1$.
Thus, the TYPE-1-credits left by \emph{all} nodes inside $\father(b_1)$
 is at least $|b_1|(k+1) - \min(k+1,x_1+x_2) \cdot (k+1)$,
 while the TYPE-1-credits contributed
 by nodes inside boxes  $b_1'$, $b_2'$, $b_3'$, $\cdots$,
 $b_i'$ is at least $x_1 (k+1)$.
Thus, the total free TYPE-1-credits in box $\father(b_1)$ is at least
 $|b_1|(k+1) - \min(k+1-x_1,x_2)(k+1)$.
Observe that, when we process the box   $\father(b_1)$, we already
  have $e_2= \min(k+1-x_1,x_2)(k+1)$ edges crossing the box
  $\father(b)$.
These edges are added from processing pairs of box $b_1$ and
 non-$k$-boxes $b_j'$, $j >i$.
Thus, our claim holds. 
\end{compactenum}

We can show  by induction that Lemma
 \ref{lemm:enough-credits} holds for all $k$-boxes.
This finishes the proof of lemma. 

\medskip

For a $k$-box $b$, assume that it already has $e_2$ edges  crossing
 the boundary of $b$ in the direction of a cone $B$.
Then, by our lemma, we know that it has at least $(k+1)^2 - e_2$
 free TYPE-1-credits.
Assume that we process  pairs of potential-edge boxes $b$ and $b_j'$, for
 $j\le t$, where $b_j'$ is a non-$k$-box.
For simplicity, we define variable $x_1$ (as number of nodes from
$b_j'$ that are inside $\father(b)$) and $x_2$ (as number of nodes from
$b_j'$ that are outside $\father(b)$) similarly.
We  only have to add at most $\min(k+1, x_1+x_2)(k+1) - e_2$ edges
 to nodes inside $b$.
Similarly, we can show that our statement (Lemma
\ref{lemm:enough-credits}) holds for the parent $k$-box $\father(b)$ also.

Observe that in our algorithm, when we decide to add edges $uv$ for a
 pair of potential-edge boxes, we choose a pair of nodes each has the
 smallest degree in the corresponding direction.
Thus, the degree difference among all nodes is at most $1$.
Since we assigned $k+1$ TYPE-1-credits to every node for adding edges in the
 case $uv$ when $u$ is from a $k$-box, and $v$ is from a non-$k$-box,
 the maximum number of edges added to a node $u$ in this case will be
 at most $k+1$ in any cone direction.

\item 
We then study the total number of edges $uv$ that are added to a node
 $u$ when $u$ is inside some $k$-box $b_1$ and $v$ is inside some
 $k$-box $b_1'$ and the boxes $b_1$ and $b_1'$ are a pair of
 potential-edge boxes.

Notice that we add at most $k+1$  edges when we processing each pair of
 potential-edge boxes $b$ and $b'$ that are $k$-boxes.
In this case, we will assign $2$ TYPE-2-credits for each node inside a
 $k$-box.
Recall that the TYPE-2-credits are different from the TYPE-1-credits
 in the previous case.
We charge a node $u$ a TYPE-2-credit  if an edge $uv$ is added where $u$ and
 $v$ are from $b$ and $b'$ respectively.
Similar to Lemma \ref{lemm:enough-credits}, we can prove that
\begin{lemma}
\label{lemm:enough-credits-k-k}
 Every $k$-box $b$ will have  at least $2(k+1) - e_2 $ free TYPE-2-credits
 where $e_2$ is the number of edges crossing the boundary of $b$ in
 the given cone direction
 when we start processing this box $b$  and its
 potential-edge boxes to add some edge after its children boxes have
 been processed.
\end{lemma}
We only added at most
 $1$ edge when a node $u$ is served as a representative node of a $k$-box
 $b_1$.
The statement clearly is true for all leaf-$k$-boxes.
Then consider a non leaf-$k$-box $b$.
If we have a $k$-box $b'$ such that $b'$ is inside $\father(b)$,
 then after processed $b$ and $b'$, $\father(b)$ will have at least
 $2(k+1)$ free TYPE-2-credits, where $b$ and $b'$ contributed $k+1$
 TYPE-2-credits  each.
Similarly we can show that the lemma is true when $b'$ is outside of
 $\father(b)$.
\end{compactenum}
Thus the  theorem follows.
\end{proof}

It is also not difficult to show the following theorem.
\begin{theorem}
\label{theorem:bound-time-k}
Algorithm \ref{alg:kfts-k2} can be implemented
 to run in time $O( k c_2 n+c_2 n \log n)$, where
 $c_2=\Theta((\frac{1}{t-1})^d)$  is the number of cones
 partitioned (which in turn is dependent on the spanning ratio $t$).
\end{theorem}
\begin{proof}
For each box $b$ and each direction $B_i$, we store $k+1$ disjoint edges
 in an array $\DisjointCrossingEdge(b,B_i)$
 such that the end-nodes are furthest from the box in this direction.
These edges $uv$ will be sorted based on the distances from $v$ to
 the box $b$, where $u$ is a node inside $b$.
For a box $b$, given the array  $\DisjointCrossingEdge(b_1,B_i)$,
 and  $\DisjointCrossingEdge(b_2,B_i)$ where $b_1$ and $b_2$ are two
 children boxes of $b$,
 we clearly can update the list for  box $b$ in $O(k)$ time from
  $2$ sorted lists from children boxes $b_1$ and $b_2$ as follows.
We greedily compare the top
 elements (the link with the furthest node) of two children boxes and
 get the link $uv$ with $v$ being furthest from $b$ in direction
 $B_i$; the process is repeated till we 
 find $k+1$ links. Notice that these newly found $k+1$ links surely
 will be disjoint also.
Here if a node $v$ is connected to multiple nodes $u_1$
 and $u_2$ for DisjointCrossingEdge arrays $\DisjointCrossingEdge(b_1,
 B)$ and  $\DisjointCrossingEdge(b_2, B)$ for both children boxes
 $b_1$ and $b_2$, we will only pick one link $u_iv$ for one of
 the children boxes and discard the other.
The total time of such processing is $2 \cdot (k+1)$.
Since there are $O(n)$ boxes in total, the total time complexity of
 updating the furthest links list $\DisjointCrossingEdge$ can be done in time
 $O( k n)$ for a single cone direction.
The time complexity then follows from the fact that there are $c_2$
 cones.
\end{proof}


\section{Conclusion}
\label{sec:conclusion}

In this paper, we studied the spanner construction for a set of $n$ points
 in $\dspace$
 and also fault-tolerant spanners for a set of points
 in $\dspace$.
Our main contribution is 
 an  algorithm that runs in time $O(n \log n)$ to
 construct a $(k,t)$-VFTS for Euclidean graph  with
 maximum node degree $O(1+k)$, and weight at most $O((1+k)^2)
 \weight(\MST)$ for $k \ge 0$. All bounds are asymptotically
 optimum.
It remains an interesting future work to extend the method to geodesic
 distance when we are given $n$ nodes on a surface.

{\small 

}

\end{document}